\documentclass[11pt,a4paper]{article}

\usepackage[T1]{fontenc}
\usepackage[utf8]{inputenc}
\usepackage{lmodern}
\usepackage{amsmath,amssymb,amsthm,mathtools,mathrsfs}
\usepackage[a4paper,margin=1.04in,includefoot]{geometry}
\usepackage{xurl}
\usepackage[hidelinks]{hyperref}
\input{glyphtounicode}
\allowdisplaybreaks
\newtheorem{theorem}{Theorem}[section]
\newtheorem{lemma}[theorem]{Lemma}
\newtheorem{proposition}[theorem]{Proposition}
\newtheorem{corollary}[theorem]{Corollary}
\theoremstyle{definition}

\theoremstyle{remark}
\newtheorem{remark}[theorem]{Remark}

\newcommand{\X}{\mathscr X_n}
\newcommand{\C}{\mathcal C}
\newcommand{\Pcode}{\mathcal P}
\newcommand{\one}{\mathbf 1}
\newcommand{\Dtwo}{D^{L_2}}
\newcommand{\Adual}{A^\bot}
\newcommand{\wh}{\widehat}
\newcommand{\E}{\mathbb E}
\newcommand{\TV}{d_{\mathrm{TV}}}
\newcommand{\wt}{\operatorname{wt}}

\hypersetup{
  pdflang={en-US},
  pdftitle={Quantitative Tiling Stability From Quadratic Discrepancy In Hamming Spaces},
  pdfauthor={Valery (Binyamin) Grishin and Aryeh Lev Zabokritskiy (Yohananov)},
  pdfkeywords={quadratic discrepancy, perfect code, stability, code smoothing,
    Hamming space, Krawtchouk polynomial}
}

\title{Quantitative Tiling Stability From Quadratic Discrepancy\\
In Hamming Spaces}

\author{%
  \large Valery (Binyamin) Grishin\\[0.3em]
  \small Department of Computer Science, Tel-Hai University of Kiryat Shmona\\
  \small and the Galilee, Kiryat Shmona, Israel\\
  \small \href{mailto:valerig.tech@gmail.com}{\texttt{valerig.tech@gmail.com}}\\[0.8em]
  \large Aryeh Lev Zabokritskiy (Yohananov)\\[0.3em]
  \small Department of Computer Science, MIGAL -- Galilee Research Institute,\\
  \small Kiryat Shmona, Israel\\
  \small Department of Computer Science, Tel-Hai University of Kiryat Shmona\\
  \small and the Galilee, Kiryat Shmona, Israel\\
  \small \href{mailto:yuhanalev@telhai.ac.il}{\texttt{yuhanalev@telhai.ac.il}}
}

\date{}

\begin{document}

\maketitle

\begin{abstract}
Quadratic ball discrepancy measures how uniformly a code meets Hamming balls
over all centers and radii.  Perfect codes minimize this quantity among codes
of the same cardinality.  We prove a quantitative refinement: at the parameters
of every perfect code other than a one-word code or the full space, the
discrepancy excess over the minimum is at least an explicit positive
multiple of the discrepancy at the packing radius.  The latter equals the tiling defect divided by the square of the code
cardinality. This defect is the total squared deviation of the covering
multiplicity from one.  The comparison
applies to all codes of that cardinality.  For one-error perfect codes, the
coefficient is at least one, and no larger constant holds for all such codes;
explicit coefficients depending on the parameters can be much larger.  The
tiling defect also measures the failure of uniform ball noise to smooth a code
distribution.
Consequently, the same discrepancy excess gives explicit bounds on holes,
overlaps, and divergence from the uniform distribution.
\end{abstract}



\section{Introduction}

Perfect codes partition Hamming space into balls of equal radius.  They also
solve a discrepancy minimization problem: among codes of the same
cardinality, they minimize the total quadratic deviation of ball counts from
their uniform values.  Barg proved the binary result using Stolarsky's
invariance principle and linear programming
\cite{Barg2021Stolarsky}.  The extension to every finite alphabet was proved
in \cite[Theorems~2.1--2.3]{Zabokritskiy2026Perfect}, with explicit lower
bounds that depend only on the parameters and with equality characterized by
perfect tiling.  The present paper asks
how much of the tiling survives when the discrepancy is close to its minimum.

By Stolarsky's invariance principle, quadratic ball discrepancy equals a
constant minus the average, over ordered pairs of codewords, of a fixed
function of their Hamming distance \cite{Barg2021Stolarsky}.  The question
therefore belongs to the broader study of energy minimization in Hamming
spaces.  Binomial moments and universal optimality identify exact energy
minimizers for classes of potentials
\cite{AshikhminBarg1999Binomial,CohnZhao2014Energy,BoyvalenkovEtAl2017}.
Weighted ball discrepancy \cite{BargSkriganov2021} and the relative deviation of
random linear codes \cite{DoronEtAl2026} give other measures of uniformity.
Here the objective is total quadratic ball discrepancy, and we seek a
quantitative bound on holes and overlaps at one radius.
For one-error parameters, this stability question was posed in
\cite[Section~7, open problem~1]{Zabokritskiy2026Perfect}.

Fix an alphabet $\mathcal A$ of size $q$, a length $n$, and a code
$\C\subseteq\mathcal A^n$ of cardinality $N$.  Write $B(x,t)$ for a Hamming ball of radius $t$,
$V_t$ for its cardinality, and
$m_{t,\C}(x)\triangleq|\C\cap B(x,t)|$ for the covering multiplicity of $x$, the
number of codewords whose radius-$t$ balls contain $x$.  The total quadratic
ball discrepancy is
\[
 \Dtwo(\C)\triangleq\sum_{t=0}^n\sum_{x\in\mathcal A^n}
 \left(\frac{m_{t,\C}(x)}{N}-\frac{V_t}{q^n}\right)^2.
\]
We compare all $N$-element subsets of $\mathcal A^n$.  Thus the minimum
distance and any linear structure of a perfect code impose no restrictions
on its competitors.

Suppose that the sphere-packing identity $NV_e=q^n$ holds.  Perfect tiling
at radius $e$ means that $m_{e,\C}(x)=1$ everywhere.  Its failure is measured by the tiling defect
\[
 \Phi_e(\C)\triangleq\sum_{x\in\mathcal A^n}\bigl(m_{e,\C}(x)-1\bigr)^2.
\]
This defect records both holes (uncovered words) and overlaps (multiple
coverage).  The
radius-$e$ contribution to $\Dtwo(\C)$ is exactly $\Phi_e(\C)/N^2$.
The issue is the contribution of the other radii: lowering their discrepancy
could compensate for a defect at radius $e$.  Exact minimization alone gives
no explicit bound on this compensation.

Our main result supplies such a bound.  If $\Pcode$ is a nontrivial perfect code
correcting $e$ errors, that is, neither a one-word code nor the full space, then
there is an explicit constant $\sigma_{n,q,e}>0$, depending only on the
parameters, such that every code $\C$ of size $N=|\Pcode|$ satisfies
\[
 \Dtwo(\C)-\Dtwo(\Pcode)
 \geq \sigma_{n,q,e}\frac{\Phi_e(\C)}{N^2}.
\]
In particular, discrepancy close to the minimum forces the covering
multiplicity to be close to one.  For one-error perfect codes, the coefficient
can always be taken to be one.  Equivalently, after the radius-one contribution is removed, the
discrepancy of $\C$ over the other radii is still at least that of $\Pcode$.
No larger constant holds for all one-error perfect codes, although explicit
coefficients depending on $n$ and $q$ can be much larger.

The same defect has a probabilistic interpretation.  Choose a uniform
codeword and then a uniform word in the radius-$e$ ball centered there.  Under $NV_e=q^n$, the chi-square
divergence of the resulting distribution from uniformity is
$\Phi_e(\C)/q^n$.  Code smoothing with ball and Bernoulli noise was studied by
Pathegama and Barg
\cite[Propositions~4, 7, and A1]{PathegamaBarg2023Smoothing};
a framework for codes and lattices appears in
\cite{DebrisAlazardEtAl2023Smoothing}.  This chi-square identity is already
part of the theory of code smoothing and of
\cite[Proposition~6.1]{Zabokritskiy2026Perfect}.  Our comparison bounds this
divergence by an explicit multiple of the discrepancy excess
$\Dtwo(\C)-\Dtwo(\Pcode)$.  It therefore gives explicit
bounds on holes, overlaps, total variation, and R\'enyi divergences of order
at most two.
These bounds control the covering multiplicities; whether $\C$ must also be
close in symmetric difference to some perfect code is a further question,
related in spirit to stability for Delsarte-tight spherical codes
\cite{BoroczkyGlazyrin2024Stability}.

\paragraph{Results and their scope}
The main theorem combines the following results.  The first two need only
arithmetic conditions on the parameters, not the existence of a perfect code.
They bound the excess of $\Dtwo(\C)$ over the lower bound of
\cite[Theorems~2.1 and~2.2]{Zabokritskiy2026Perfect}, a benchmark computed
from the parameters alone: the discrepancy that every perfect code with these
parameters has whenever such a code exists.
\begin{enumerate}
\item For one-error parameters with $n\geq3$, assume that
$N=q^n/(1+n(q-1))$ and the Lloyd root $(1+n(q-1))/q$ are integers.
The excess is at least an explicit multiple of $\Phi_1(\C)/N^2$ with a
parameter-dependent coefficient that is always at least one; it is also at most an explicit multiple of
$\Phi_1(\C)/N^2$, and this upper coefficient can be much larger than the
lower one.
\item For two-error parameters with $n\geq5$, assume that
$N=q^n/V_2$ is an integer and that the two-error Lloyd polynomial has two
distinct integral roots in $\{1,\ldots,n\}$.  The excess is at least an
explicit positive multiple of $\Phi_2(\C)/N^2$; for $q\geq4$ the coefficient
has a closed form.
\item Stability bounds for odd binary repetition codes and the binary Golay
parameters supply the remaining comparisons.  Together with the
classification and nonexistence results
\cite{Tietavainen1973,ZinovievLeontiev1973,Reuvers1977,Best1983,Hong1984},
the three items give the main theorem for every nontrivial perfect code.
\end{enumerate}
For $q=2,3$, the two-error conditions hold only for binary length five and
ternary length eleven \cite[Lemma~5.6]{Zabokritskiy2026Perfect}, where perfect
codes exist; no two-error parameter
set with $q\geq4$ satisfying these conditions is known.  The existence of perfect
two-error codes over non-prime-power alphabets remains open in general,
despite strong arithmetic exclusions \cite{CazorlaGarcia2024,Bennett2026}.

We build on the exact minimization theorems of
\cite{Zabokritskiy2026Perfect} and the results used in their proofs.  The new
results are the stability inequalities, their explicit coefficients, and
their consequences for holes, overlaps, and smoothing.

\paragraph{Organization}
Section~\ref{sec:defect-main} relates the tiling defect to holes, overlaps, and
smoothing and states the main theorem.  Section~\ref{sec:spectral-inputs}
collects the spectral formulas for the discrepancy and the tiling defect.
Sections~\ref{sec:one-error} and~\ref{sec:higher-error} establish the one- and
two-error comparisons and the bound for binary repetition codes.
Section~\ref{sec:remaining-perfect} treats the binary Golay parameters and
completes the proof of the main theorem.  Section~\ref{sec:consequences}
derives the smoothing bounds and bounds on codeword pairs at distance one or
two, and discusses the multiplicity profile of nearly-perfect covering codes.

\section{The tiling defect and the main results}
\label{sec:defect-main}

We first isolate the tiling defect controlled by the discrepancy,
identify it with a smoothing error, and then state the main theorem.
For the spectral arguments, choose a finite abelian group $G$ of order $q$ and
a bijection from $\mathcal A$ to $G$.  Its coordinatewise extension is a
Hamming isometry, so this relabeling preserves every discrepancy and tiling
quantity in the paper.  The group structure is used only for convolution and Fourier analysis.
We may therefore write $\X\triangleq G^n$ without restricting the alphabet size,
and we let $d$ denote Hamming distance on $\X$.
For $0\leq t\leq n$, write
\[
 B(x,t)\triangleq\{y\in\X:d(x,y)\leq t\},\qquad
 V_t\triangleq|B(x,t)|=\sum_{j=0}^t\binom nj(q-1)^j.
\]
For a nonempty code $\C\subseteq\X$ of size $N$, its total quadratic ball
discrepancy is
\begin{equation}
\label{eq:discrepancy}
 \Dtwo(\C)\triangleq\sum_{t=0}^n\sum_{x\in\X}
 \left(\frac{|\C\cap B(x,t)|}{N}-\frac{V_t}{q^n}\right)^2.
\end{equation}
This normalization agrees with \cite{Barg2021Stolarsky,Zabokritskiy2026Perfect}.
For fixed parameters, write
\begin{equation}
\label{eq:extremal-value}
 \Dtwo(q,n,N)\triangleq
 \min_{\substack{\C\subseteq G^n\\|\C|=N}}\Dtwo(\C).
\end{equation}

Fix an integer $e$ with $1\leq e\leq n$ and assume throughout this section that
\begin{equation}
\label{eq:perfect-cardinality}
 NV_e=q^n.
\end{equation}
The code $\C$ is otherwise arbitrary.  The covering multiplicity and the
tiling defect introduced above are
\begin{equation}
\label{eq:defect}
 m_{e,\C}(x)\triangleq|\C\cap B(x,e)|,
 \qquad
 \Phi_e(\C)\triangleq\sum_{x\in\X}\bigl(m_{e,\C}(x)-1\bigr)^2.
\end{equation}
Let
\[
 H_e(\C)\triangleq|\{x:m_{e,\C}(x)=0\}|
\]
be the number of holes, and let
\[
 Z_e(\C)\triangleq|\{x:m_{e,\C}(x)\neq1\}|
\]
be the number of defective ambient points.  For $S\subseteq\X$, $\one_S$
denotes its indicator function.  We use natural logarithms and the conventions
$\TV(P,Q)\triangleq\frac12\sum_x|P(x)-Q(x)|$ and
$\chi^2(P\|Q)\triangleq\sum_x(P(x)-Q(x))^2/Q(x)$ when $Q$ has full support.
The symbols $D(P\|Q)$ and $D_2(P\|Q)$ denote relative entropy and order-two
R\'enyi divergence.

The overlap identity below is the fixed-radius form of
\cite[Theorem~2.1 and Eq.~(8)]{Barg2021Stolarsky}.  The output-law and
chi-square identities are the $q$-ary form of the binary convolution formulas in
\cite[Propositions~4 and A1]{PathegamaBarg2023Smoothing}, as in
\cite[Proposition~6.1]{Zabokritskiy2026Perfect}.  We record them together with
the total-variation identity and the hole decomposition needed for the
stability estimates.

\begin{lemma}[Defect, overlaps, and smoothing]
\label{lem:defect-identities}
Under \eqref{eq:perfect-cardinality},
\begin{align}
 \Phi_e(\C)
 &=\sum_{\substack{z,z'\in\C\\z\neq z'}}
   |B(z,e)\cap B(z',e)|,\label{eq:overlap-identity}\\
 \Phi_e(\C)
 &=2H_e(\C)+
   \sum_{x:m_{e,\C}(x)\geq2}
   (m_{e,\C}(x)-1)(m_{e,\C}(x)-2).
   \label{eq:hole-decomposition}
\end{align}
Consequently,
\begin{equation}
\label{eq:hole-defective}
 2H_e(\C)\leq\Phi_e(\C),
 \qquad
 Z_e(\C)\leq\Phi_e(\C).
\end{equation}

Let $X$ be uniform on $\C$, let $E$ be independent of $X$ and uniform on
$B(0,e)$, and put $Y\triangleq X+E$.  If $P_{e,\C}$ is the law of $Y$ and $U$ is
uniform on $\X$, then
\begin{align}
 P_{e,\C}(x)&=\frac{m_{e,\C}(x)}{q^n},
 \label{eq:output-law}\\
 \chi^2(P_{e,\C}\|U)&=\frac{\Phi_e(\C)}{q^n},
 \label{eq:chi-defect}\\
 \TV(P_{e,\C},U)&=\frac{H_e(\C)}{q^n}.
 \label{eq:tv-holes}
\end{align}
Moreover,
\begin{equation}
\label{eq:defect-discrete-range}
 \Phi_e(\C)=0\qquad\text{or}\qquad \Phi_e(\C)\geq4.
\end{equation}
The defect $\Phi_e(\C)$ vanishes if and only if $\C$ is a perfect $e$-error-correcting code.
\end{lemma}

\begin{proof}
Equation \eqref{eq:perfect-cardinality} gives
$\sum_xm_{e,\C}(x)=q^n$.  Therefore
\[
 \Phi_e(\C)=\sum_xm_{e,\C}(x)^2-q^n.
\]
Double counting triples $(x,z,z')$ with
$x\in B(z,e)\cap B(z',e)$ shows that the right side is precisely the sum of
the ordered overlaps with $z\neq z'$, proving
\eqref{eq:overlap-identity}.

The deviations $m_{e,\C}(x)-1$ sum to zero.  Every negative deviation is
exactly $-1$, so
\[
 \sum_{x:m_{e,\C}(x)\geq2}(m_{e,\C}(x)-1)=H_e(\C).
\]
Expanding each positive square as $r^2=r+r(r-1)$ gives
\eqref{eq:hole-decomposition}.  The first inequality in
\eqref{eq:hole-defective} follows.  The second follows pointwise from
$\one_{\{m\neq1\}}\leq(m-1)^2$ for every nonnegative integer $m$.

The convolution of the uniform measures on $\C$ and $B(0,e)$ gives
\eqref{eq:output-law}.  Since $U(x)=q^{-n}$,
\[
 \chi^2(P_{e,\C}\|U)
 =q^n\sum_x\left(\frac{m_{e,\C}(x)-1}{q^n}\right)^2,
\]
which is \eqref{eq:chi-defect}.  The total negative deviation is
$H_e(\C)$ and equals the total positive deviation, so the $\ell_1$ distance
is $2H_e(\C)/q^n$, proving \eqref{eq:tv-holes}.  Finally, if $\Phi_e(\C)$ is
positive, \eqref{eq:overlap-identity} gives distinct centers whose
radius-$e$ balls meet.  Their intersection contains at least two points: if
the distance between the centers is at most $e$, the intersection contains
both centers; if this distance lies strictly between $e$ and $2e$, two
successive points on a shortest path lie in both balls; and if it equals $2e$,
there are at least two midpoints.  The ordered
center pair and its reverse therefore contribute at least four in total.
This proves \eqref{eq:defect-discrete-range}.  Vanishing is exactly the
perfect-tiling identity.
\end{proof}

The next proposition translates every stability estimate below into
combinatorial bounds and bounds on the divergences just defined.

\begin{proposition}[From a stability inequality to smoothing]
\label{prop:transfer}
Suppose that a benchmark $\delta$ depending only on the parameters and a constant $\sigma>0$
satisfy
\begin{equation}
\label{eq:abstract-stability}
 \Dtwo(\C)-\delta\geq\sigma\frac{\Phi_e(\C)}{N^2}
\end{equation}
for every $N$-word code, where $NV_e=q^n$.  Put
$\Delta(\C)\triangleq\Dtwo(\C)-\delta$.  Then
\begin{equation}
\label{eq:overlap-transfer}
 \sum_{\substack{z,z'\in\C\\z\neq z'}}
 |B(z,e)\cap B(z',e)|
 =\Phi_e(\C)\leq\frac{N^2}{\sigma}\Delta(\C).
\end{equation}
Moreover,
\begin{align}
 H_e(\C)&\leq\frac{N^2}{2\sigma}\Delta(\C),
 &Z_e(\C)&\leq\frac{N^2}{\sigma}\Delta(\C),
 \label{eq:count-transfer}\\
 \chi^2(P_{e,\C}\|U)&\leq
 \frac{N}{\sigma V_e}\Delta(\C),
 &\TV(P_{e,\C},U)&\leq
 \frac{N}{2\sigma V_e}\Delta(\C).
 \label{eq:smoothing-transfer}
\end{align}
Also,
\begin{equation}
\label{eq:renyi-transfer}
 D(P_{e,\C}\|U)
 \leq D_2(P_{e,\C}\|U)
 \leq\log\left(1+\frac{N}{\sigma V_e}\Delta(\C)\right).
\end{equation}
If $\C$ is not perfect, then
\begin{equation}
\label{eq:discrete-gap-abstract}
 \Delta(\C)\geq\frac{4\sigma}{N^2}.
\end{equation}
\end{proposition}

\begin{proof}
Equations \eqref{eq:overlap-transfer}--\eqref{eq:smoothing-transfer} follow
from Lemma~\ref{lem:defect-identities} and
\eqref{eq:abstract-stability}, using $q^n=NV_e$.  Moreover,
\[
 D_2(P\|U)=\log\bigl(1+\chi^2(P\|U)\bigr).
\]
This identity and monotonicity of R\'enyi divergence in its order give
\eqref{eq:renyi-transfer}.  If $\C$ is not perfect, then
\eqref{eq:defect-discrete-range} gives $\Phi_e(\C)\geq4$, proving
\eqref{eq:discrete-gap-abstract}.
\end{proof}

We call a perfect code nontrivial if it is neither the full space correcting
zero errors nor a one-word code; in particular, the binary repetition codes of
odd length at least three are nontrivial.  Theorem~2.3 of
\cite{Zabokritskiy2026Perfect} established the corresponding exact equality
characterization; the theorem below re-establishes it as a consequence of
stronger stability inequalities.  Its
ingredients are as follows: Theorem~\ref{thm:one-uniform} treats
one-error parameters, Corollary~\ref{cor:all-q-two-stability} treats two-error
parameters, Theorem~\ref{thm:repetition} treats odd binary repetition codes,
and Proposition~\ref{prop:binary-golay-stability} treats the binary Golay
parameters.  The first two results require only arithmetic admissibility; the
classification is used at the end to exhaust the actual perfect codes.

\begin{theorem}[Global tiling-defect stability]
\label{thm:global-stability}
Let $G$ be a finite abelian group of order $q$ and let
$\Pcode\subseteq G^n$ be a nontrivial perfect code correcting $e$ errors.
Put $N\triangleq|\Pcode|$.  There is an explicit constant
$\sigma_{n,q,e}>0$, depending only on $n$, $q$, and $e$, such that every $N$-word code
$\C\subseteq G^n$ satisfies
\begin{equation}
\label{eq:global-stability}
 \Dtwo(\C)-\Dtwo(\Pcode)
 \geq\sigma_{n,q,e}\frac{\Phi_e(\C)}{N^2}.
\end{equation}
The proof supplies the following choices:
\begin{enumerate}
\item $\sigma_{n,q,1}\triangleq1$ for every one-error perfect code;
\item for every two-error perfect code, take the positive coefficient stated
in Corollary~\ref{cor:all-q-two-stability};
\item $\sigma_{2e+1,2,e}\triangleq1$ for binary repetition codes with $e\geq3$;
\item $\sigma_{23,2,3}\triangleq4136/25$ for the binary Golay code.
\end{enumerate}
Consequently, all conclusions of Proposition~\ref{prop:transfer} hold with
$\delta=\Dtwo(\Pcode)$.  Moreover,
$\Dtwo(\C)=\Dtwo(\Pcode)$ if and only if $\C$ is itself a perfect
$e$-error-correcting code.
\end{theorem}

For a one-error perfect code $\Pcode$, the coefficient-one case reads
\[
 \Dtwo(\C)-\frac{\Phi_1(\C)}{N^2}\geq\Dtwo(\Pcode).
\]
The subtracted term is exactly the radius-one summand of $\Dtwo(\C)$, while it
vanishes for $\Pcode$.  Thus the other radii alone still favor the
perfect-code benchmark.  The conclusion concerns the covering multiplicities
of $\C$; proximity of $\C$ to the particular code $\Pcode$ in symmetric
difference is a separate question.

\section{Spectral inputs}
\label{sec:spectral-inputs}

The stability inequalities compare discrepancy excess with the fixed-radius
tiling defect.  Fourier expansion writes total discrepancy and the defect as
sums over the same nonnegative spectral coordinates, weighted respectively
by discrepancy spectral weights and squared ball coefficients.  After subtracting
the benchmark, we express the discrepancy excess in these same coordinates
using the relevant moment constraints or formal reference vector.  The
one-error and Golay arguments compare the resulting weights, and summing over
these coordinates gives the desired inequality for the original code; the
two-error argument for $q\geq4$ and the repetition-code argument work in
distance coordinates instead.
We record only the parts of the
Fourier--Krawtchouk theory needed for this comparison.  Standard background on
the Hamming association scheme and finite Fourier analysis can be found in
\cite{Delsarte1973,MacWilliamsSloane1977,Terras1999}.

For $0\leq j,w\leq n$, the $q$-ary Krawtchouk polynomial in our normalization
is
\[
 K_j^{(n,q)}(w)\triangleq
 \sum_{\ell=0}^j(-1)^\ell(q-1)^{j-\ell}
 \binom w\ell\binom{n-w}{j-\ell}.
\]
The distance distribution of an $N$-word code is
\[
 A_i(\C)\triangleq\frac1N|\{(z,z')\in\C^2:d(z,z')=i\}|,
 \qquad 0\leq i\leq n.
\]
Thus $A_0=1$ and $\sum_iA_i=N$.  Its MacWilliams--Delsarte transform is
\begin{equation}
\label{eq:dual-distribution}
 \Adual_w(\C)\triangleq
 \frac1N\sum_{i=0}^nA_i(\C)K_w^{(n,q)}(i),
 \qquad 0\leq w\leq n.
\end{equation}
Let $\widehat G$ be the character group of $G$.  A frequency
$\xi=(\xi_1,\ldots,\xi_n)\in\widehat G^{\,n}$ determines the character
$\chi_\xi(x)\triangleq\prod_i\xi_i(x_i)$.  If $\gamma_0$ is the trivial
character, set
\[
 \wt(\xi)\triangleq|\{i:\xi_i\neq\gamma_0\}|.
\]
We use the unnormalized Fourier transform
\[
 \wh f(\xi)\triangleq\sum_{x\in\X}f(x)\overline{\chi_\xi(x)}.
\]
For nonlinear codes, $\Adual$ is a transform rather than the distribution of
a dual code.  It is nonnegative and has the Fourier expression
\[
 \Adual_w(\C)=\frac1{N^2}
 \sum_{\substack{\xi\in\widehat G^{\,n}\\\wt(\xi)=w}}
 |\wh{\one_\C}(\xi)|^2.
\]

Define the Fourier coefficient of a radius-$t$ ball and the discrepancy
spectral weight by
\begin{equation}
\label{eq:c-W}
 c_w(t)\triangleq\sum_{j=0}^tK_j^{(n,q)}(w),
 \qquad
 W_{n,q}(w)\triangleq\sum_{t=0}^{n-1}c_w(t)^2.
\end{equation}
The Fourier and moment identities below are standard in the Hamming
association scheme and were also used, in the present normalization, in
\cite{Zabokritskiy2026Perfect}.
For functions $f,g:\X\to\mathbb C$, we use the unnormalized convolution
\[
 (f*g)(x)\triangleq\sum_{y\in\X}f(y)g(x-y).
\]
With the Fourier normalization above,
\[
 \wh{f*g}(\xi)=\wh f(\xi)\wh g(\xi),
 \qquad
 \sum_{x\in\X}|f(x)|^2
 =\frac1{q^n}\sum_{\xi\in\widehat G^{\,n}}|\wh f(\xi)|^2.
\]

\begin{lemma}[Fourier transform of a Hamming ball]
\label{lem:ball-transform}
For $1\leq w\leq n$ and $0\leq t\leq n-1$,
\[
 \wh{\one_{B(0,t)}}(\xi)
 =c_w(t)
 =K_t^{(n-1,q)}(w-1)
 \qquad\text{whenever }\wt(\xi)=w.
\]
Moreover, $c_w(n)=0$.
\end{lemma}

\noindent Its proof is reproduced in
Appendix~\ref{app:reproduced-prerequisites}.

\begin{proposition}[Spectral Parseval identities]
\label{prop:spectral-parseval}
Every nonempty code $\C\subseteq G^n$ satisfies
\begin{equation}
\label{eq:spectral-discrepancy}
 \Dtwo(\C)=\frac1{q^n}\sum_{w=1}^n
 \Adual_w(\C)W_{n,q}(w).
\end{equation}
If $NV_e=q^n$, then
\begin{equation}
\label{eq:defect-parseval}
 \frac{\Phi_e(\C)}{N^2}
 =\frac1{q^n}\sum_{w=1}^n
 \Adual_w(\C)c_w(e)^2.
\end{equation}
\end{proposition}

\noindent The proof of \eqref{eq:spectral-discrepancy} is reproduced in
Appendix~\ref{app:reproduced-prerequisites}.  The fixed-radius identity
\eqref{eq:defect-parseval} connects the inherited spectral framework with the
tiling defect; its binary case appears in
\cite[Proposition~4]{PathegamaBarg2023Smoothing}.
\begin{proof}[Proof of \eqref{eq:defect-parseval}]
Put
\[
 h_e\triangleq
 \one_\C*\one_{B(0,e)}-\one_{\X}.
\]
The hypothesis $NV_e=q^n$ makes its trivial Fourier coefficient zero, while
\[
 \wh h_e(\xi)
 =\wh{\one_\C}(\xi)c_{\wt(\xi)}(e)
 \qquad(\xi\neq0).
\]
Since $h_e(x)=m_{e,\C}(x)-1$, Plancherel and the definition of
$\Adual_w$ give \eqref{eq:defect-parseval}.
\end{proof}

For a real polynomial $P(z)=\sum_jp_jz^j$, write
\[
 \|P\|_2^2\triangleq\sum_j|p_j|^2,
\]
and put
\[
 u(z)\triangleq1+(q-1)z,\qquad
 b(z)\triangleq1-z,\qquad
 R_q(z)\triangleq z^2+(q-2)z+1.
\]

\begin{proposition}[Coefficient norm and spectral convexity]
\label{prop:W-shape}
For $1\leq w\leq n$,
\begin{equation}
\label{eq:W-coefficient-norm}
 W_{n,q}(w)
 =\left\|u(z)^{n-w}b(z)^{w-1}\right\|_2^2.
\end{equation}
If $n\geq3$, then, for $1\leq w\leq n-2$,
\begin{equation}
\label{eq:curvature-norm}
\begin{aligned}
 &W_{n,q}(w+2)-2W_{n,q}(w+1)+W_{n,q}(w)\\
 &\qquad
 =q^2\left\|
 u(z)^{n-w-2}b(z)^{w-1}R_q(z)
 \right\|_2^2>0.
\end{aligned}
\end{equation}
Thus $W_{n,q}(1),\ldots,W_{n,q}(n)$ is strictly discretely convex.
For $q\geq4$, it is also strictly decreasing.
\end{proposition}

\noindent Its proof is reproduced in
Appendix~\ref{app:reproduced-prerequisites}.

The one-error comparison in Section~\ref{sec:one-error} assumes $n\geq3$ and
the integrality conditions \eqref{eq:one-parameters}.  It needs a strict local
decrease of $W_{n,q}$ at the one-error Lloyd root $k\triangleq(1+n(q-1))/q$,
which is then an integer with $2\leq k\leq n$:
\begin{equation}
\label{eq:local-slope}
 W_{n,q}(k)<W_{n,q}(k-1).
\end{equation}
For $q\geq4$ this follows from Proposition~\ref{prop:W-shape}.  For $q=2$
it follows from the closed formula for $W_{n,2}$ derived in the proof of
Proposition~\ref{prop:clean-spectral-gap} (the binary formula of
\cite[Equation~(37)]{Barg2021Stolarsky}); for $q=3$ it follows from
Lemmas~\ref{lem:ternary-coefficient} and~\ref{lem:ternary-recurrence} in
Appendix~\ref{app:ternary-slope}.

For $j\geq0$, write
\[
 (x)_j\triangleq x(x-1)\cdots(x-j+1),
 \qquad (x)_0\triangleq1,
\]
for the falling factorial.

\begin{lemma}[Factorial dual moments]
\label{lem:dual-factorial-moments}
For every $N$-word code $\C\subseteq G^n$,
\begin{equation}
\label{eq:dual-mass}
 \sum_{w=1}^n\Adual_w=\frac{q^n}{N}-1.
\end{equation}
Moreover, for $1\leq j\leq n$,
\begin{equation}
\label{eq:dual-factorial-moments}
 \sum_{w=1}^n(w)_j\Adual_w
 =\frac{q^{n-j}}N
 \sum_{i=0}^j(-1)^i
 (j)_i(n-i)_{j-i}
 (q-1)^{j-i}A_i.
\end{equation}
\end{lemma}

\noindent Its proof is reproduced in
Appendix~\ref{app:reproduced-prerequisites}.

\begin{proposition}[Lloyd support {\cite{Zabokritskiy2026Perfect}}]
\label{prop:lloyd-support}
Let $1\leq e\leq n-1$, and let $\Pcode\subseteq G^n$ be a perfect
$e$-error-correcting code.  Then
\[
 \Adual_w(\Pcode)=0
 \qquad\text{whenever}\qquad
 c_w(e)=K_e^{(n-1,q)}(w-1)\neq0
 \qquad(1\leq w\leq n).
\]
Thus the nonconstant dual spectrum of $\Pcode$, that is, $\Adual_w(\Pcode)$ for
$w\geq1$, is supported on the integral roots of its Lloyd polynomial
$w\mapsto c_w(e)$.
\end{proposition}

\begin{proof}
The radius-$e$ balls centered at the words of $\Pcode$ partition $\X$, so
$|\Pcode|V_e=q^n$ and $\Phi_e(\Pcode)=0$.  Hence \eqref{eq:defect-parseval},
applied to $\C=\Pcode$, gives $\sum_{w=1}^n\Adual_w(\Pcode)c_w(e)^2=0$.  Every
summand is nonnegative, so $\Adual_w(\Pcode)=0$ whenever $c_w(e)\neq0$.
\end{proof}

We will also need the one-error moments.

\begin{lemma}[One-error moments]
\label{lem:one-error-moments}
Assume $V\triangleq1+n(q-1)$, $N\triangleq q^n/V$, and $k\triangleq V/q$ are integers.  For every
$N$-word code,
\begin{align}
 \sum_{w=1}^n\Adual_w&=V-1,
 \label{eq:mass-moment}\\
 \sum_{w=1}^n(w-k)\Adual_w&=-kA_1,
 \label{eq:first-centered}\\
 \sum_{w=1}^n(w-k)^2\Adual_w
 &=kA_1+\frac{2k}{q}A_2.
 \label{eq:second-centered}
\end{align}
Moreover,
\begin{equation}
\label{eq:phi-one-distance}
 \Phi_1(\C)=N(qA_1+2A_2),
\end{equation}
and hence
\begin{equation}
\label{eq:B-phi}
 \sum_{w=1}^n(w-k)^2\Adual_w
 =\frac{k}{qN}\Phi_1(\C).
\end{equation}
\end{lemma}

\begin{proof}
The inherited moment calculation proving
\eqref{eq:mass-moment}--\eqref{eq:second-centered} is reproduced in
Appendix~\ref{app:reproduced-prerequisites}.  It remains to identify the
quadratic moment with the tiling defect.  Two distinct radius-one balls meet in
$q$ points when their centers are at distance one, in two points when their
centers are at distance two, and not at larger distance.  Therefore
\eqref{eq:overlap-identity} gives \eqref{eq:phi-one-distance}; combining it with
\eqref{eq:second-centered} proves \eqref{eq:B-phi}.
\end{proof}

\section{One-error stability}
\label{sec:one-error}

Let $G$ be a finite abelian group of order $q$.  Assume throughout this
section that $q\geq2$, $n\geq3$, and
\begin{equation}
\label{eq:one-parameters}
 V\triangleq1+n(q-1),\qquad N\triangleq q^n/V,\qquad k\triangleq V/q
\end{equation}
are integers.  Put $W(w)\triangleq W_{n,q}(w)$ and
\begin{equation}
\label{eq:one-benchmark}
 \delta_{n,q}^{(1)}\triangleq\frac{V-1}{q^n}W(k).
\end{equation}
We call such a pair $(q,n)$ an admissible one-error parameter set.

We first state the uniform bound.  Its proof combines a sharper
parameter-dependent comparison with a lower bound for the coefficient of that comparison.

\begin{theorem}[Sharp uniform one-error stability]
\label{thm:one-uniform}
Under \eqref{eq:one-parameters}, every $N$-word code
$\C\subseteq G^n$ satisfies
\begin{equation}
\label{eq:one-uniform}
 \Dtwo(\C)-\delta_{n,q}^{(1)}
 \geq\frac{\Phi_1(\C)}{N^2}.
\end{equation}
The benchmark $\delta_{n,q}^{(1)}$ is attained exactly when $\C$ is a
perfect one-error-correcting code.  If no perfect one-error code exists, then
\begin{equation}
\label{eq:one-nonattainment}
 \Dtwo(q,n,N)
 \geq\delta_{n,q}^{(1)}+\frac4{N^2}.
\end{equation}
The coefficient one in \eqref{eq:one-uniform} is best possible uniformly
over the admissible parameter sets.
\end{theorem}

Since
$c_w(1)=1+n(q-1)-qw=q(k-w)$, the defect formula
\eqref{eq:defect-parseval} is a centered second spectral moment.  We therefore
seek a quadratic minorant of $W$ that touches it at the two adjacent weights
$k-1$ and $k$.

Let
\[
 a_{n,q}\triangleq W(k-1)-W(k),
\]
which is positive by \eqref{eq:local-slope}, and define
\begin{equation}
\label{eq:eta-definition}
 \eta_{n,q}\triangleq
 \min_{\substack{1\leq w\leq n\\w\notin\{k-1,k\}}}
 \frac{W(w)-W(k)+a_{n,q}(w-k)}{(w-k)(w-k+1)}.
\end{equation}
Strict convexity makes every ratio in \eqref{eq:eta-definition} positive.
Set
\begin{equation}
\label{eq:kappa-definition}
 \kappa_{n,q}\triangleq\min\{a_{n,q},\eta_{n,q}\}.
\end{equation}
For the upper comparison, put
\begin{equation}
\label{eq:rho-definition}
 \rho_{n,q}^{(1)}\triangleq
 \max_{\substack{1\leq w\leq n\\w\neq k}}
 \frac{W(w)-W(k)}{q^2(w-k)^2}.
\end{equation}

\begin{theorem}[Parameter-dependent one-error comparison]
\label{thm:one-two-sided}
Under \eqref{eq:one-parameters}, every $N$-word code
$\C\subseteq G^n$ satisfies
\begin{equation}
\label{eq:one-two-sided}
 \frac{\kappa_{n,q}}{q^2N^2}\Phi_1(\C)
 \leq \Dtwo(\C)-\delta_{n,q}^{(1)}
 \leq \rho_{n,q}^{(1)}\frac{\Phi_1(\C)}{N^2}.
\end{equation}
In particular, the benchmark is attained if and only if $\C$ is a perfect
one-error-correcting code.
\end{theorem}

\begin{proof}
Write $t\triangleq w-k$.  The definition of $\kappa_{n,q}$ gives the
pointwise minorant
\begin{equation}
\label{eq:quadratic-minorant}
 W(w)\geq
 W(k)-a_{n,q}t+\kappa_{n,q}t(t+1).
\end{equation}
At $t=-1,0$ this is an equality; away from those two integers,
$t(t+1)>0$ and \eqref{eq:quadratic-minorant} follows from
$\kappa_{n,q}\leq\eta_{n,q}$ and \eqref{eq:eta-definition}.

Put $Q_w\triangleq\Adual_w(\C)$ and
\[
 B\triangleq\sum_{w=1}^n(w-k)^2Q_w.
\]
By \eqref{eq:spectral-discrepancy}, \eqref{eq:mass-moment}, and
\eqref{eq:one-benchmark},
\begin{equation}
\label{eq:one-gap-spectral}
 q^n\bigl(\Dtwo(\C)-\delta_{n,q}^{(1)}\bigr)
 =\sum_{w=1}^nQ_w\bigl(W(w)-W(k)\bigr).
\end{equation}
Summing \eqref{eq:quadratic-minorant} and using
\eqref{eq:first-centered} gives
\[
 q^n\bigl(\Dtwo(\C)-\delta_{n,q}^{(1)}\bigr)
 \geq \kappa_{n,q}B
 +(a_{n,q}-\kappa_{n,q})kA_1
 \geq\kappa_{n,q}B.
\]
For the other direction, \eqref{eq:rho-definition} gives directly
\[
 q^n\bigl(\Dtwo(\C)-\delta_{n,q}^{(1)}\bigr)
 \leq \rho_{n,q}^{(1)}q^2B.
\]
Now use \eqref{eq:B-phi} and $q^n=Nqk$ to obtain
\eqref{eq:one-two-sided}.  Its outer terms vanish simultaneously exactly when
$\Phi_1(\C)=0$, which is perfect tiling by
Lemma~\ref{lem:defect-identities}.
\end{proof}

The lower coefficient in Theorem~\ref{thm:one-two-sided} satisfies the
parameter-free estimate $\kappa_{n,q}/q^2\geq1$.  This is the step that makes the final stability
constant independent of the alphabet.  To estimate it, we express $W$ as a
coefficient norm: the first difference at $k$ controls $a_{n,q}$, while the
second differences control $\eta_{n,q}$.  The two estimates then combine
through the definition of $\kappa_{n,q}$.  Proposition~\ref{prop:W-shape}
supplies both the coefficient norm and the required curvature identity.
\begin{proposition}[One-error coefficient]
\label{prop:clean-spectral-gap}
Let $q\geq2$, $r\geq1$, $n=qr+1$, and
$k=1+r(q-1)$.  Define
\begin{equation}
\label{eq:gamma-definition}
 \beta_{n,q}\triangleq\min_{1\leq j\leq n-2}
 \bigl(W(j+2)-2W(j+1)+W(j)\bigr).
\end{equation}
Then
\begin{equation}
\label{eq:clean-coefficient-bounds}
 a_{n,q}\geq q^2,
 \qquad
 \beta_{n,q}\geq2q^2,
 \qquad
 \kappa_{n,q}\geq q^2.
\end{equation}
If $q\geq4$, the minimum in \eqref{eq:eta-definition} is attained uniquely
at $w=n$, and
\begin{equation}
\label{eq:kappa-closed-qge4}
 \begin{aligned}
 \kappa_{n,q}=\eta_{n,q}
 &=\frac{W(n)+rW(k-1)-(r+1)W(k)}{r(r+1)}\\
 &=\frac{q^2}{r(r+1)}
   \sum_{j=0}^{r-1}(j+1)
   \left\|u(z)^j b(z)^{qr-j-2}R_q(z)\right\|_2^2.
 \end{aligned}
\end{equation}
These inequalities do not require $N=q^n/V$ to be an integer.
\end{proposition}

\begin{proof}
For $1\leq w\leq n-2$, the polynomial inside the norm in
\eqref{eq:curvature-norm} has nonzero
coefficients at two distinct degrees: its constant and leading terms have
absolute values $1$ and
$(q-1)^{n-w-2}$.  Hence its squared norm is at least two, proving
$\beta_{n,q}\geq2q^2$.

It remains to bound the local slope.  The relations between $n,k,r$ give
$n-k=r$ and $k-2=r(q-1)-1$.  Set $d\triangleq qr-1$ and
\[
 F(z)\triangleq u(z)^rb(z)^{r(q-1)-1}
 =\sum_{j=0}^df_jz^j.
\]
Then
\begin{equation}
\label{eq:slope-norm-difference}
\begin{aligned}
 a_{n,q}
 &=\|uF\|_2^2-\|bF\|_2^2\\
 &=q\left((q-2)\sum_{j=0}^df_j^2
       +2\sum_{j=0}^{d-1}f_jf_{j+1}\right).
\end{aligned}
\end{equation}
Applying $2xy\geq-x^2-y^2$ to adjacent coefficients gives
\[
 \frac{a_{n,q}}q
 \geq(q-4)\|F\|_2^2+f_0^2+f_d^2
 =(q-4)\|F\|_2^2+1+(q-1)^{2r}.
\]
For $q\geq4$, this is at least $q$, and therefore $a_{n,q}\geq q^2$.

For $q=2$, write $n=2k-1$.  Parseval's identity turns
\eqref{eq:W-coefficient-norm} into
\[
 W_{n,q}(w)=\frac1{2\pi}\int_0^{2\pi}|u(e^{i\theta})|^{2(n-w)}
 |b(e^{i\theta})|^{2(w-1)}\,d\theta,
\]
and for $q=2$,
$|u(e^{i\theta})|^2=4\cos^2(\theta/2)$ and
$|b(e^{i\theta})|^2=4\sin^2(\theta/2)$.  The beta integral
\[
 \frac1{2\pi}\int_0^{2\pi}\sin^{2w-2}(\theta/2)\cos^{2n-2w}(\theta/2)\,d\theta
 =\frac{\binom{2w-2}{w-1}\binom{2n-2w}{n-w}}{4^{n-1}\binom{n-1}{w-1}}
 \qquad(1\leq w\leq n)
\]
therefore gives
$W_{n,2}(w)=\binom{2n-2w}{n-w}\binom{2w-2}{w-1}/\binom{n-1}{w-1}$.  By
Lemma~\ref{lem:ball-transform},
$W_{n,2}(w)=\sum_{t=0}^{n-1}K_t^{(n-1,2)}(w-1)^2$, so this is also the closed
formula \cite[Equation~(37)]{Barg2021Stolarsky} for sums of squared binary
Krawtchouk values; hence
\[
 a_{n,2}=\frac{2}{2k-3}\binom{2k-2}{k-1}\geq4,
\]
because the expression equals $4$ at $k=2$ and the ratio of its values at
$k+1$ and $k$ is $2(2k-3)/k\geq1$.  For $q=3$, we have $k=2r+1$, so
$a_{3r+1,3}=W(2r)-W(2r+1)$; with the integer $\gamma_r$ defined in
Appendix~\ref{app:ternary-slope},
Lemmas~\ref{lem:ternary-coefficient} and~\ref{lem:ternary-recurrence} give
\[
 a_{3r+1,3}=6\gamma_r\geq12>9
 \qquad(r\geq1).
\]

Finally, write $d_j\triangleq W(j+1)-W(j)$.  The definition \eqref{eq:gamma-definition} says
$d_{j+1}-d_j\geq\beta_{n,q}$.  If $w=k+t$ with $t\geq1$, then
\[
\begin{aligned}
 W(w)-W(k)+a_{n,q}t
 &=\sum_{i=0}^{t-1}(d_{k+i}-d_{k-1})\\
 &\geq\beta_{n,q}\sum_{i=1}^{t}i
 =\frac{\beta_{n,q}}2t(t+1).
\end{aligned}
\]
If $w=k-1-s$ with $s\geq1$, then
\[
\begin{aligned}
 W(w)-W(k)-a_{n,q}(s+1)
 &=\sum_{i=0}^{s}(d_{k-1}-d_{k-1-i})\\
 &\geq\beta_{n,q}\sum_{i=1}^{s}i
 =\frac{\beta_{n,q}}2s(s+1).
\end{aligned}
\]
These are the numerators of the ratios in \eqref{eq:eta-definition} at
$w=k+t$ and $w=k-1-s$, whose denominators are $t(t+1)$ and $s(s+1)$, so $\eta_{n,q}\geq\beta_{n,q}/2$.  Together
with \eqref{eq:kappa-definition} and the two bounds already proved, this
establishes \eqref{eq:clean-coefficient-bounds}.

It remains to prove the exact assertion when $q\geq4$.  If
$P(z)=\sum_{i=0}^dp_iz^i\in\mathbb R[z]$ has nonzero constant and leading coefficients,
then
\begin{equation}
\label{eq:norm-monotonicity}
 \begin{aligned}
 \|uP\|_2^2-\|bP\|_2^2
 &=q\left((q-2)\|P\|_2^2
       +2\sum_{i=0}^{d-1}p_ip_{i+1}\right)\\
 &\geq q\left((q-4)\|P\|_2^2+p_0^2+p_d^2\right)>0.
 \end{aligned}
\end{equation}
Applying this first to
$P=u^{n-w-1}b^{w-1}$ shows that $W(w)$ is strictly decreasing.  Applying it
to $P=u^{n-j-3}b^{j-1}R_q$ in \eqref{eq:curvature-norm} shows that
\[
 c_j\triangleq W(j+2)-2W(j+1)+W(j)
\]
is strictly decreasing in $j$.

For $1\leq t\leq r=n-k$, the ratio in
\eqref{eq:eta-definition} at $w=k+t$ is
\begin{equation}
\label{eq:right-eta-average}
 R_t^+\triangleq
 \frac{\sum_{h=0}^{t-1}(t-h)c_{k-1+h}}{t(t+1)}.
\end{equation}
These ratios decrease strictly with $t$.  Indeed, with
$x_h\triangleq c_{k-1+h}$, $N_t\triangleq\sum_{h=0}^{t-1}(t-h)x_h$, and
$S_t\triangleq\sum_{h=0}^tx_h$, we have
\[
 2N_t-tS_t
 =\sum_{\ell=1}^t\ell(t-\ell+1)(x_{\ell-1}-x_\ell)>0,
\]
which is equivalent to $R_{t+1}^+<R_t^+$.  On the other side, the ratio at
$w=k-1-s$ is
\[
 R_s^-\triangleq
 \frac{\sum_{h=0}^{s-1}(s-h)c_{k-2-h}}{s(s+1)}.
\]
Each of $R_s^-$ and $R_t^+$ is one half of a convex combination of the
corresponding curvature values.  Since every curvature index in $R_s^-$ lies
strictly to the left of every curvature index in $R_t^+$ and $c_j$ decreases
strictly, $R_s^->R_t^+$ for all $1\leq s\leq k-2$ and $1\leq t\leq r$.  Hence the unique minimum
is $R_r^+$, attained at
$w=n$.  Since $W(n)<W(k)$, this minimum is smaller than $a_{n,q}$, proving
the first line of \eqref{eq:kappa-closed-qge4}.  Finally, substitute
\eqref{eq:curvature-norm} into \eqref{eq:right-eta-average} with $t=r$ and
change variables $j=r-1-h$ to obtain the second line.
\end{proof}

\begin{table}[htbp]
\caption{Lower and upper coefficients of Theorem~\ref{thm:one-two-sided} at
the two admissible one-error parameter sets of smallest length for each $q\in\{2,3,4\}$.  Entries written as integers are exact; the others are rounded to three significant digits.}
\label{tab:one-error-coefficients}
\centering
\begin{tabular}{c|c|c|c}
\hline
$(q,n)$ & $\kappa_{n,q}/q^2$ & $\rho_{n,q}^{(1)}$
& $\rho_{n,q}^{(1)}/(\kappa_{n,q}/q^2)$\tabularnewline
\hline
$(2,3)$  & $1$ & $1$ & $1$\tabularnewline
$(2,7)$  & $2$ & $2.51\mathbin{\times}10^1$ & $1.26\mathbin{\times}10^1$\tabularnewline
$(3,4)$  & $1$ & $6.42$ & $6.42$\tabularnewline
$(3,13)$ & $1.98\mathbin{\times}10^3$
         & $8.41\mathbin{\times}10^7$ & $4.24\mathbin{\times}10^4$\tabularnewline
$(4,5)$  & $3$ & $1.47\mathbin{\times}10^2$ & $4.90\mathbin{\times}10^1$\tabularnewline
$(4,21)$ & $7.40\mathbin{\times}10^8$
         & $4.88\mathbin{\times}10^{19}$ & $6.60\mathbin{\times}10^{10}$\tabularnewline
\hline
\end{tabular}
\end{table}

In Table~\ref{tab:one-error-coefficients}, the uniform coefficient one can be much smaller than the
parameter-dependent lower coefficient, and the lower and upper coefficients can
differ by many orders of magnitude.

\begin{proof}[Proof of Theorem~\ref{thm:one-uniform}]
Under \eqref{eq:one-parameters}, $qk=1+n(q-1)$ gives $n\equiv1\pmod q$, so
$n=qr+1$ with $r\geq1$ (as $n\geq3$) and $k=1+r(q-1)$.  Hence
Theorem~\ref{thm:one-two-sided} and
Proposition~\ref{prop:clean-spectral-gap} prove
\eqref{eq:one-uniform}.  If no perfect code exists, \eqref{eq:defect-discrete-range}
gives $\Phi_1\geq4$ for every code, proving \eqref{eq:one-nonattainment}.  Uniform sharpness
follows from either of the two examples below.
\end{proof}

At
$(q,n,N,k)=(2,3,2,2)$,
\[
 (W(1),W(2),W(3))=(6,2,6),
 \qquad
 \kappa_{3,2}=4,
 \qquad
 \rho_{3,2}^{(1)}=1.
\]
Every nonperfect two-word code has discrepancy excess one and $\Phi_1=4$, so
both inequalities in \eqref{eq:one-two-sided} are equalities.

Sharpness also occurs at a nonbinary perfect-code parameter set.  Over
$\mathbb F_3$, let
\[
 \C\triangleq\{(a,a,b,b):a,b\in\mathbb F_3\}.
\]
This nonperfect $[4,2,2]_3$ code has
\[
 A(\C)=\Adual(\C)=(1,0,4,0,4),
 \qquad
 (W(1),W(2),W(3),W(4))=(245,26,14,20).
\]
Consequently,
\[
 \Phi_1(\C)=72,\qquad
 \Dtwo(\C)-\delta_{4,3}^{(1)}
 =\frac89=\frac{\Phi_1(\C)}{9^2}.
\]
Thus the uniform coefficient one is attained by a nonperfect code also at the ternary
Hamming parameters $(q,n,N)=(3,4,9)$.

\section{Two-error stability}
\label{sec:higher-error}

We
first state the conditional large-alphabet estimate, then construct the formal
reference vector and prove the coefficient inequalities needed for it.
The arithmetic hypotheses reduce $q=2,3$ to two explicit cases, giving an
all-alphabet theorem at the end of the section.

\subsection{The conditional large-alphabet estimate}
\label{sec:two-error}

Put
\begin{equation}
\label{eq:two-parameters}
 V_2=1+n(q-1)+\binom n2(q-1)^2,
 \qquad N\triangleq q^n/V_2.
\end{equation}
Assume $q\geq2$, $n\geq5$, $N$ is an integer, and the Lloyd polynomial
\[
 L_2(w)\triangleq c_w(2)=K_2^{(n-1,q)}(w-1)
\]
has two distinct integral roots $r<s$ in $\{1,\ldots,n\}$; we then call $(q,n)$ an admissible two-error parameter set.  The notation $L_2$ is unrelated to the superscript $L_2$ in $\Dtwo$, which denotes quadratic discrepancy.  Let
\[
 \bar w\triangleq\frac{(q-1)n}{q},
\]
\begin{equation}
\label{eq:formal-masses}
 b_r\triangleq\frac{(V_2-1)s-V_2\bar w}{s-r},
 \qquad
 b_s\triangleq\frac{V_2\bar w-(V_2-1)r}{s-r},
\end{equation}
and define the formal discrepancy benchmark
\begin{equation}
\label{eq:two-benchmark}
 \delta_{n,q}^{(2)}\triangleq
 \frac{b_rW_{n,q}(r)+b_sW_{n,q}(s)}{q^n}.
\end{equation}
These formal data depend only on $q$ and $n$: $b_r$ and $b_s$ solve
$b_r+b_s=V_2-1$ and $rb_r+sb_s=V_2\bar w$, which by
Lemma~\ref{lem:dual-factorial-moments} are the values of
$\sum_{w\geq1}\Adual_w$ and $\sum_ww\Adual_w$ for every $N$-word code with
$A_1=0$.

Put
\[
 m\triangleq n-1,
 \qquad a\triangleq q-4,
\]
and define
\begin{align}
\label{eq:P-four}
 \psi_4(m,a)\triangleq{}&
 m(m-1)(m-2)(m-3)a^4
 +8m(m-1)(m-2)a^3\notag\\
 &+72m(m-1)a^2+480ma+1680.
\end{align}
For $q\geq4$, the stability coefficient in the next theorem is
\begin{equation}
\label{eq:vartheta-closed}
 \vartheta_{n,q}\triangleq
 \frac{\binom{2m-8}{m-4}\psi_4(m,a)}
 {6m(m-1)(m-2)(m-3)}.
\end{equation}

\begin{theorem}[Large-alphabet two-error stability]
\label{thm:two-stability}
Let $G$ be a finite abelian group of order $q\geq4$ and let $n\geq5$.
Assume that $V_2$ divides $q^n$ and that
$L_2(w)=K_2^{(n-1,q)}(w-1)$ has two distinct integral roots in
$\{1,\ldots,n\}$.  Put $N\triangleq q^n/V_2$.  Then every code
$\C\subseteq G^n$ with $|\C|=N$ satisfies
\begin{equation}
\label{eq:two-stability}
 \Dtwo(\C)-\delta_{n,q}^{(2)}
 \geq\vartheta_{n,q}\frac{\Phi_2(\C)}{N^2}.
\end{equation}
If a perfect two-error code exists, every perfect code with these parameters
attains the benchmark.
If no perfect two-error code exists, then
\begin{equation}
\label{eq:two-gap}
 \Dtwo(q,n,N)\geq\delta_{n,q}^{(2)}+\frac{4\vartheta_{n,q}}{N^2}.
\end{equation}
\end{theorem}

\begin{remark}[Status of the large-alphabet branch]
No nontrivial perfect two-error code is currently known for $q\geq4$.
Prime-power alphabets are excluded by the classical classification, while
non-prime-power alphabets remain open under strong arithmetic restrictions
\cite{CazorlaGarcia2024,Bennett2026}.

The arithmetic conditions are sparse even before the known nonexistence
theorems are applied.  An exact integer search \cite{Zabokritskiy2026StabilityVerifier} over
$2\leq q\leq300$ and $5\leq n\leq4000$ found that both conditions hold only
for $(q,n)=(2,5)$ and $(3,11)$.  Among the $123$ pairs with $q\geq4$ in this
range for which the Lloyd polynomial has two integral roots, none satisfies
$V_2\mid q^n$.  The search used the exact quadratic identity
\[
 2L_2(w)=q^2w^2-q\bigl(2(q-1)(n-2)+3q\bigr)w+2V_2
\]
to test the discriminant, root integrality, and root range, followed by exact
modular exponentiation for the divisibility condition.  The conclusion of
this calculation is restricted to the searched rectangle.
\end{remark}

\begin{proposition}[Explicit large-alphabet coefficient]
\label{prop:vartheta-closed}
For every pair of integers $q\geq4$ and $n\geq5$, the number
$\vartheta_{n,q}$ in \eqref{eq:vartheta-closed} is positive.  For each fixed
$q$, this analytic expression satisfies, as $n\to\infty$,
\begin{equation}
\label{eq:vartheta-asymptotics}
 \vartheta_{n,q}\sim
 \begin{cases}
 \displaystyle
 \frac{(q-4)^4}{6144\sqrt\pi}\frac{4^n}{n^{1/2}},&q\geq5,\\[2mm]
 \displaystyle
 \frac{35}{128\sqrt\pi}\frac{4^n}{n^{9/2}},&q=4.
 \end{cases}
\end{equation}
The limit is taken over all integers $n\geq5$, with $q$ fixed.
\end{proposition}

\begin{proof}
Positivity is immediate from \eqref{eq:P-four}, and Stirling's formula
applied to \eqref{eq:vartheta-closed} gives \eqref{eq:vartheta-asymptotics};
Appendix~\ref{app:vartheta-evaluation} gives the details and proves
\eqref{eq:F4-vartheta}.
\end{proof}

For centers at distance $w$, let $\mu_2(w)$ be the intersection size of their
radius-two balls.  Unlike the one-error proof, the argument for $q\geq4$ does
not bound $W$ pointwise.  It works in distance coordinates, where $\mu_2$
vanishes beyond distance four (see \eqref{eq:mu-two}) and the discrepancy
potential is completely monotonic in the sense defined below
(Proposition~\ref{prop:distance-complete-monotone}).  We compare the actual distance distribution with a formal
reference obtained from dual mass at the two Lloyd roots.  Once transformed
back to distance coordinates, this formal reference has zero entries at
distances $1,\ldots,4$, as for a perfect two-error code, although the
reference need not be realizable.  The formal comparison
shows that, for every completely monotonic function, the sum against the actual
distance distribution is at least the sum against the reference.  The new step
shows that, for the discrepancy potential, this difference is at least
$\vartheta_{n,q}$ times the corresponding difference for $\mu_2$; the two
differences equal $N(\Dtwo(\C)-\delta_{n,q}^{(2)})$ and $\Phi_2(\C)/N$,
which gives Theorem~\ref{thm:two-stability}.  A Newton expansion about $w=n$
handles the binomial moments of order at least five, while the first four are
regrouped by the distances $1,\ldots,4$.

Define the distance kernel
\[
 \lambda(x,y)\triangleq\frac12\sum_{u\in\X}|d(x,u)-d(y,u)|.
\]
This depends only on $d(x,y)$; write $\lambda(w)$ for its value at distance
$w$.  For a function $u:\{0,\ldots,n\}\to\mathbb R$, put
$\Delta u(w)\triangleq u(w+1)-u(w)$ for $0\leq w<n$.  We call $u$
\emph{completely monotonic} if
\[
 (-1)^j\Delta^ju(w)\geq0
 \qquad(0\leq w\leq n,\quad 0\leq j\leq n-w).
\]
Define the discrepancy potential
\begin{equation}
\label{eq:f-potential}
 f_{n,q}(w)\triangleq\lambda(n)-\lambda(w).
\end{equation}
Also set
\[
 \Lambda_{n,q}\triangleq
 \frac1{q^n}\sum_{i=0}^n\binom ni(q-1)^i\lambda(i).
\]

The next proposition appeared in
\cite{Zabokritskiy2026Perfect}.
\begin{proposition}[Distance representation and complete monotonicity]
\label{prop:distance-complete-monotone}
For all $n\geq1$ and $N\geq1$, every $N$-word code $\C\subseteq G^n$ satisfies
\begin{equation}
\label{eq:distance-form}
 \Dtwo(\C)=\Lambda_{n,q}
 -\frac1N\sum_{i=1}^n A_i(\C)\lambda(i).
\end{equation}
If $q\geq4$, the function $f_{n,q}$ is completely monotonic, and every
positive-order alternating difference is strictly positive.
\end{proposition}

\noindent Its proof is reproduced in
Appendix~\ref{app:reproduced-prerequisites}.

Until the end of the proof of Theorem~\ref{thm:two-stability}, assume $q\geq4$ and write $f\triangleq f_{n,q}$.
Define the formal dual vector $Q^*$ by
\[
 Q_0^*\triangleq1,\qquad Q_r^*\triangleq b_r,\qquad Q_s^*\triangleq b_s,
 \qquad Q_w^*\triangleq0\quad(w\notin\{0,r,s\}).
\]
 Let $\mathsf K\triangleq(K_i^{(n,q)}(j))_{i,j=0}^n$ be the Krawtchouk matrix, so that
$\mathsf K^2=q^nI$, and put
\begin{equation}
\label{eq:formal-A-star}
 A^*\triangleq\frac1{V_2}\mathsf KQ^*.
\end{equation}
Thus $A^*$ is a formal reference vector; it is not assumed to be nonnegative
or to arise from a code.

The following comparison applies the pair-covering certificate of
\cite[Definition~20 and Proposition~21]{CohnZhao2014Energy} to the formal
reference $A^*$; its two-node quadrature is
the strength-four Rao-endpoint case of Levenshtein's quadrature
\cite[Section~2.5]{BoyvalenkovEtAl2017}.  In this form it appears as
Lemma~5.3 and Lemma~A.1 of \cite{Zabokritskiy2026Perfect}.  We state it in
the present normalization; its proof is in
Appendix~\ref{app:formal-comparison}.

\begin{lemma}[Formal two-error quadrature and comparison]
\label{lem:formal-two-comparison}
Under the assumptions of Theorem~\ref{thm:two-stability}, the formal data above
satisfy
\[
 1<r<\bar w<s<n,\qquad b_r,b_s>0,
\]
and
\begin{equation}
\label{eq:A-star-low}
 A_0^*=1,\qquad A_1^*=A_2^*=A_3^*=A_4^*=0,
 \qquad \sum_{i=0}^nA_i^*=N.
\end{equation}
If $A$ is the distance distribution of any $N$-word code, then every
completely monotonic function $u:\{0,\ldots,n\}\to\mathbb R$ satisfies
\begin{equation}
\label{eq:formal-completely-monotone-comparison}
 u^{\mathsf T}(A-A^*)\geq0.
\end{equation}
 The comparison permits negative entries $A_i^*$ for $i\geq5$.
\end{lemma}

Comparing \eqref{eq:spectral-discrepancy} with \eqref{eq:distance-form} for
a one-word code, and then for a two-word code with distance $i$, gives the
kernel identity
\[
 \lambda(i)=\Lambda_{n,q}
 -\frac1{q^n}\sum_{j=1}^nW_{n,q}(j)K_j^{(n,q)}(i).
\]
Moreover, \eqref{eq:formal-A-star} and $\mathsf K^2=q^nI$ imply
$Q^*=N^{-1}\mathsf KA^*$ because $NV_2=q^n$.  Summing the kernel identity
against $A_i^*/N$ over $0\leq i\leq n$, with $\sum_iA_i^*=N$ from
\eqref{eq:A-star-low} and this formal transform, gives
\begin{equation}
\label{eq:formal-benchmark-distance}
 \delta_{n,q}^{(2)}=
 \Lambda_{n,q}-\frac1N\sum_{i=1}^nA_i^*\lambda(i).
\end{equation}
Fix an arbitrary $N$-word code $\C$ and write $A\triangleq A(\C)$ for its distance
distribution.  The quantitative step that remains is to prove
\begin{equation}
\label{eq:formal-comparison-strong}
 f^{\mathsf T}(A-A^*)
 \geq\vartheta_{n,q}\mu_2^{\mathsf T}(A-A^*).
\end{equation}
The two identities that connect this formal comparison with the quantities in
the theorem are
\begin{align}
\label{eq:f-gap}
 f^{\mathsf T}(A-A^*)
 &=N\bigl(\Dtwo(\C)-\delta_{n,q}^{(2)}\bigr),\\
\label{eq:mu-gap}
 \mu_2^{\mathsf T}(A-A^*)&=\Phi_2(\C)/N.
\end{align}
They are verified at the end of the proof of Theorem~\ref{thm:two-stability} below.  Thus
\eqref{eq:formal-comparison-strong} immediately yields
\eqref{eq:two-stability}.

A direct coordinate count gives, for $n\geq5$,
\begin{equation}
\label{eq:mu-two}
\begin{aligned}
 \mu_2(0)&=V_2,\\
 \mu_2(1)&=q\bigl(1+(n-1)(q-1)\bigr),\\
 \mu_2(2)&=q^2+2(n-2)(q-1),\\
 \mu_2(3)&=6(q-1),\\
 \mu_2(4)&=6,
\end{aligned}
\qquad
 \mu_2(w)=0\quad(w\geq5).
\end{equation}
For $w=4,3,2,1$, direct counting gives, respectively,
$6$, $6+6(q-2)$, $q^2+2(n-2)(q-1)$, and
$q+q(n-1)(q-1)$: separate the choices supported on the differing
coordinates from those using one common coordinate.  The cases $w=0$ and
$w\geq5$ are immediate.

For a function $u$ on $\{0,\ldots,n\}$, write
\[
 \mathcal D_{j,w}u\triangleq(-1)^j\Delta^ju(w),
 \qquad 0\leq j\leq n,\quad 0\leq w\leq n-j.
\]
For $1\leq t\leq n-1$, define
\[
 F_t\triangleq\mathcal D_{n-t,t}f,
 \qquad
 M_t\triangleq\mathcal D_{n-t,t}\mu_2.
\]
The Newton expansion about $w=n$
\begin{equation}
\label{eq:right-newton}
 u(w)=\sum_{t=0}^n
 \mathcal D_{n-t,t}u\binom{n-w}{n-t},
\end{equation}
with an out-of-range binomial coefficient interpreted as zero, follows by
iterating the difference recurrence and applying Pascal's identity; compare
\cite[Lemma~4]{CohnZhao2014Energy}.  A short normalization check is included
in Appendix~\ref{app:vartheta-evaluation}.  The same
appendix evaluates $F_4$ and proves
\begin{equation}
\label{eq:F4-vartheta}
 F_4=6\vartheta_{n,q}.
\end{equation}
The integral-root hypothesis forces $n\geq q+3$ when $q\geq4$;
the short arithmetic argument is included in
Appendix~\ref{app:vartheta-evaluation}.  The proof below groups the first four
binomial moments of $A-A^*$ by the
nonnegative low-distance coordinates $A_1,\ldots,A_4$.

\begin{proof}[Proof of Theorem~\ref{thm:two-stability}]
Continue with the arbitrary code $\C$ and its distance distribution $A$ fixed
above.  Set
\[
 c\triangleq\vartheta_{n,q}=F_4/6,
 \qquad G_t\triangleq F_t-cM_t
 \quad(1\leq t\leq n-1),
\]
and define the binomial moments of $A-A^*$
\[
 B_t\triangleq
 \sum_{w=0}^n\binom{n-w}{n-t}(A_w-A_w^*).
\]
For the binomial functions $w\mapsto\binom{n-w}{n-t}$,
\[
 (-1)^j\Delta^j\binom{n-w}{n-t}
 =\binom{n-w-j}{n-t-j}\geq0,
\]
with an out-of-range binomial coefficient interpreted as zero.  Hence these
functions are completely monotonic, and
Lemma~\ref{lem:formal-two-comparison} gives $B_t\geq0$.
Since $B_0=A_0-A_0^*=0$ and $B_n=\sum_w(A_w-A_w^*)=0$, the Newton expansion
\eqref{eq:right-newton} gives
\begin{equation}
\label{eq:formal-perturbation}
 (f-c\mu_2)^{\mathsf T}(A-A^*)
 =\sum_{t=1}^{n-1}G_tB_t.
\end{equation}
For $1\leq t\leq4$, \eqref{eq:A-star-low} gives
\begin{equation}
\label{eq:low-binomial-moments}
 B_t=\sum_{i=1}^t\binom{n-i}{t-i}A_i.
\end{equation}

Direct substitution of \eqref{eq:mu-two} gives
\begin{equation}
\label{eq:M-diagonal}
\begin{aligned}
 M_1&=(m-1)^2(q-m),\\
 M_2&=q^2-4(m-1)(q-1)+3(m-1)(m-2),\\
 M_3&=6(q+1-m),\qquad M_4=6,
 \qquad M_t=0\quad(t\geq5).
\end{aligned}
\end{equation}
Since $M_4=6$ and $F_4=6c$, we have $G_4=0$.  This is the largest value of
$c$ allowed by the coefficientwise regrouping below: the $A_4$-coefficient is
$F_4-6c$.  Strict complete monotonicity of $f$ gives $F_t>0$.  Moreover,
$n\geq q+3$ implies $M_3=6(q+1-m)\leq-6<0$, so $G_3>0$; and
$G_t=F_t>0$ for $5\leq t\leq n-1$.

Also $M_1<0$, so $G_1>0$.  With $\nu\triangleq m-q\geq2$, however,
$M_2=q(2\nu-1)+(3\nu-2)(\nu-1)>0$, so a termwise bound would require
$c\leq F_2/M_2$.  We therefore group the first four terms of
\eqref{eq:formal-perturbation} by the nonnegative coordinates
$A_1,\ldots,A_4$.  For $1\leq i\leq4$, define
\[
 H_i\triangleq
 \sum_{t=i}^4\binom{n-i}{t-i}G_t.
\]
Equation \eqref{eq:low-binomial-moments} then gives
\[
 \sum_{t=1}^4G_tB_t
 =H_1A_1+H_2A_2+H_3A_3+H_4A_4.
\]
Since $G_4=0$,
\[
 \begin{aligned}
 H_4&=0, & H_3&=G_3,\\
 H_2&=G_2+(m-1)G_3,
 &H_1&=G_1+mG_2+\binom m2G_3.
 \end{aligned}
\]
Substitution from \eqref{eq:M-diagonal} gives
\begin{align*}
 T_2&\triangleq M_2+(m-1)M_3
 =-(3\nu-4)(\nu-1)-q(4\nu-5)<0,\\
 2T_1&\triangleq2M_1+2mM_2+m(m-1)M_3\\
 &=-2\nu(\nu-1)(\nu-2)-2q(3\nu^2-5\nu+1)
   -4q^2(\nu-1)<0.
\end{align*}
Therefore
\[
 H_2=F_2+(m-1)F_3-cT_2>0,
 \qquad
 H_1=F_1+mF_2+\binom m2F_3-cT_1>0.
\]
Since every $A_i$ is nonnegative and every remaining $B_t$ is nonnegative,
\eqref{eq:formal-perturbation} proves
\eqref{eq:formal-comparison-strong}.

The distance representation \eqref{eq:distance-form}, equality of the total
masses, and \eqref{eq:formal-benchmark-distance} give \eqref{eq:f-gap}.  Double
counting ordered pairs of radius-two balls gives
\[
 \Phi_2(\C)=N\bigl(\mu_2^{\mathsf T}A-V_2\bigr).
\]
By \eqref{eq:A-star-low} and the support of $\mu_2$ in
\eqref{eq:mu-two}, $\mu_2^{\mathsf T}A^*=V_2$, which proves
\eqref{eq:mu-gap}.
Combining \eqref{eq:formal-comparison-strong}--\eqref{eq:mu-gap} proves
\eqref{eq:two-stability}.  If $\Pcode$ is perfect, then
Proposition~\ref{prop:lloyd-support} shows that its nonconstant dual
spectrum is supported on the two zeros $r,s$ of $L_2$.  Equations~\eqref{eq:dual-mass} and \eqref{eq:dual-factorial-moments} with $j=1$,
together with $A_1(\Pcode)=0$, give
\[
 \Adual_r(\Pcode)+\Adual_s(\Pcode)=V_2-1,
 \qquad
 r\Adual_r(\Pcode)+s\Adual_s(\Pcode)=V_2\bar w.
\]
Thus
\[
 \Adual_r(\Pcode)=b_r,
 \qquad \Adual_s(\Pcode)=b_s.
\]
The spectral discrepancy formula now gives
$\Dtwo(\Pcode)=\delta_{n,q}^{(2)}$.  If no perfect code exists,
\eqref{eq:defect-discrete-range} gives $\Phi_2\geq4$, proving
\eqref{eq:two-gap}.
\end{proof}

Thus \eqref{eq:formal-comparison-strong} gives a quantitative refinement of
universal-optimality comparisons
\cite{AshikhminBarg1999Binomial,CohnZhao2014Energy} that applies whether or not
a code attains the formal benchmark.

The arithmetic reduction in the following lemma also appeared in
\cite{Zabokritskiy2026Perfect}; its proof is reproduced in
Appendix~\ref{app:reproduced-prerequisites}.
\begin{lemma}[The two smallest alphabets]
\label{lem:small-q-two-parameters}
Let $q\in\{2,3\}$ and $n\geq5$.  Assume that $N\triangleq q^n/V_2$ is an integer and
that $L_2$ has two distinct integral roots $r<s$ in
$\{1,\ldots,n\}$.  Then
\[
 (q,n,N,r,s)=(2,5,2,2,4)
 \quad\text{or}\quad
 (3,11,3^6,6,9).
\]
\end{lemma}

For $q=2,3$, the potential $f_{n,q}$ is not completely monotonic: the
increment formula $\lambda(w+1)-\lambda(w)=q^{n-1}p_w$ in the proof of
Proposition~\ref{prop:distance-complete-monotone} (Appendix~\ref{app:reproduced-prerequisites}),
with $p_1=0$, $p_2=1/2$ for $q=2$ and $p_1=p_2=1/3$, $p_3=7/27$ for $q=3$,
gives
$(-1)^2\Delta^2f_{n,2}(1)=-2^{n-2}$ and $(-1)^3\Delta^3f_{n,3}(1)=-2\cdot3^{n-4}$.
Lemma~\ref{lem:small-q-two-parameters} shows that the arithmetic hypotheses
leave only the binary length-five and ternary Golay parameters, which we
treat next.

\subsection{The binary case and repetition codes}

Let $e\geq1$, put $n\triangleq2e+1$, and let
$\mathcal R_n\triangleq\{0^n,1^n\}\subseteq\{0,1\}^n$.  Up to isometry, a two-word binary code
is determined by the distance $d$ between its words.  Write
$\mu_t(d)$ for the intersection size of two radius-$t$ balls whose centers are at distance
$d$.

The inequality $\Dtwo(\C)\geq\Dtwo(\mathcal R_n)$ for two-word
codes is a case of the binary theorem of \cite{Barg2021Stolarsky}; see also
\cite[Proposition~5.5]{Zabokritskiy2026Perfect}.
Keevash and Long prove an exact nonnegative decrement formula for these
intersection sizes \cite[Lemma~4.11]{KeevashLong2020Vertex}.  We include a
direct injection proof of the monotonicity needed here.

\begin{lemma}[Monotone binary ball intersections]
\label{lem:binary-intersection-monotone}
For fixed $n$ and every $0\leq t\leq n$, the sequence
$(\mu_t(d))_{d=0}^n$ is nonincreasing in $d$.
\end{lemma}

\begin{proof}
Let $y\triangleq1^d0^{n-d}$ and $y'\triangleq1^{d+1}0^{n-d-1}$, which differ
only in coordinate $d+1$.  A word $x\in B(0^n,t)\cap B(y',t)$ lies in
$B(0^n,t)\cap B(y,t)$ unless $x_{d+1}=1$ and $x$ is at distance exactly $t$
from $y'$.  In that case, change $x_{d+1}$ to zero; the new word lies in
$B(0^n,t)\cap B(y,t)$ and at distance $t+1$ from $y'$, so it is not the image
of an unchanged word.  This defines an injection from $B(0^n,t)\cap B(y',t)$
into $B(0^n,t)\cap B(y,t)$.
\end{proof}

\begin{theorem}[Exact repetition coefficient]
\label{thm:repetition}
Every two-word binary code $\C\subseteq\{0,1\}^{2e+1}$ satisfies
\begin{equation}
\label{eq:repetition-stability}
 \Dtwo(\C)-\Dtwo(\mathcal R_n)
 \geq\frac{\Phi_e(\C)}{N^2},
 \qquad N=2.
\end{equation}
The coefficient one is best possible for every odd $n$.
\end{theorem}

\begin{proof}
The threshold form of the distance kernel is
\begin{equation}
\label{eq:lambda-intersections}
 \lambda(w)=\sum_{t=0}^n\bigl(V_t-\mu_t(w)\bigr).
\end{equation}
Indeed, $|a-b|$ counts the thresholds that separate $a$ and $b$, and two
radius-$t$ balls have the same size.  By
Lemma~\ref{lem:binary-intersection-monotone},
\[
 \lambda(n)-\lambda(d)
 =\sum_{t=0}^n\bigl(\mu_t(d)-\mu_t(n)\bigr)
 \geq\mu_e(d),
\]
because the two radius-$e$ balls at antipodal centers are disjoint.
The distance representation and the overlap identity give
\[
 \Dtwo(\C)-\Dtwo(\mathcal R_n)
 =\frac{\lambda(n)-\lambda(d)}2,
 \qquad
 \Phi_e(\C)=2\mu_e(d).
\]
Since $N=2$, this proves \eqref{eq:repetition-stability}.

Equality holds at both $d=n-2$ and $d=n-1$.  Indeed, for either distance and
$t<e$, the two radius-$t$ balls are disjoint, so
$\mu_t(d)=\mu_t(n)=0$.  For $e<t<n$, inclusion--exclusion and complementation
reduce $\mu_t(d)-\mu_t(n)$ to the corresponding difference for balls of radius
$n-t-1<e$ at the same center distances; these smaller balls are again
disjoint.  At $t=n$, both balls are the whole space, so
$\mu_n(d)=\mu_n(n)=2^n$.  Thus only the term $t=e$ contributes to
$\lambda(n)-\lambda(d)$.  Direct counting gives
\[
 \mu_e(n-2)=\mu_e(n-1)=\binom{2e}{e},
\]
and hence
\[
 \lambda(n)-\lambda(d)=\mu_e(d)
 \qquad(d\in\{n-2,n-1\}).
\]
Therefore \eqref{eq:repetition-stability} is an equality at both distances,
and the coefficient cannot be increased.
\end{proof}

\subsection{The ternary case and the all-alphabet theorem}

The following proposition quantitatively refines the exact minimization result
at the ternary Golay parameters.  The unperturbed quadratic $p_0$ below appeared in
\cite[Appendix~D]{Zabokritskiy2026Perfect}; we reproduce its coefficients,
all moment evaluations, and the complete finite comparison.

\begin{proposition}[Ternary Golay stability]
\label{prop:ternary-golay-stability}
Every code $\C\subseteq\mathbb F_3^{11}$ of size $N=3^6$ satisfies
\begin{equation}
\label{eq:ternary-golay-stability}
\Dtwo(\C)-\frac{7\,446\,692}{3^{11}}
 \geq\frac{274}{9}\frac{\Phi_2(\C)}{N^2}.
\end{equation}
Every perfect two-error-correcting code attains the benchmark on the left.
Consequently,
\[
 \min_{\substack{\C\subseteq\mathbb F_3^{11}\\|\C|=3^6}}\Dtwo(\C)
 =\frac{7\,446\,692}{3^{11}},
\]
and equality in this minimum holds precisely for perfect
two-error-correcting codes.
\end{proposition}

\begin{proof}
At these parameters,
\[
 V_2=1+2\binom{11}{1}+4\binom{11}{2}=243,
 \qquad
 c_w(2)=\frac92(w-6)(w-9).
\]
Define
\[
 p_0(w)\triangleq
 323\,962-76\,236w+5\,463w(w-1),
\qquad
 \tau\triangleq\frac{274}{9}.
\]
The quadratic $p_0$ agrees with $W_{11,3}$ at $w=6,7,9$.  Since
$c_7(2)=-9\neq0$, the inequality \eqref{eq:ternary-pointwise} below cannot
hold at $w=7$ with $p_0$ in place of $p_\tau$ and a positive coefficient.  We
therefore perturb the quadratic by
\begin{equation}
\label{eq:ternary-perturbed-polynomial}
\begin{aligned}
 p_\tau(w)
 &\triangleq p_0(w)+1\,233(w-6)(w-9)\\
 &=390\,544-93\,498w+6\,696w(w-1).
\end{aligned}
\end{equation}
The exact finite calculation in
Appendix~\ref{app:finite-certificates} gives
\begin{equation}
\label{eq:ternary-pointwise}
 W_{11,3}(w)-p_\tau(w)
 \geq\tau c_w(2)^2
 \qquad(1\leq w\leq11).
\end{equation}

Equations \eqref{eq:dual-mass} and
\eqref{eq:dual-factorial-moments}, specialized to
$n=11,q=3,N=3^6$, give
\begin{align*}
 M_0&\triangleq\sum_{w=1}^{11}\Adual_w=242,\\
 M_1&\triangleq\sum_{w=1}^{11}w\Adual_w
      =1782-81A_1,\\
 M_2&\triangleq\sum_{w=1}^{11}w(w-1)\Adual_w
      =11880-1080A_1+54A_2.
\end{align*}
Hence \eqref{eq:ternary-perturbed-polynomial} yields
\begin{equation}
\label{eq:ternary-moment-perturbed}
\begin{aligned}
 \sum_{w=1}^{11}\Adual_wp_\tau(w)
 &=390\,544M_0-93\,498M_1+6\,696M_2\\
 &=7\,446\,692+341\,658A_1+361\,584A_2\\
 &\geq7\,446\,692.
\end{aligned}
\end{equation}
Summing \eqref{eq:ternary-pointwise} against $\Adual_w$ and using
\eqref{eq:spectral-discrepancy},
\eqref{eq:defect-parseval}, and
\eqref{eq:ternary-moment-perturbed} proves
\eqref{eq:ternary-golay-stability}.

If $\Pcode$ is perfect, then $A_1(\Pcode)=A_2(\Pcode)=0$, so
\eqref{eq:ternary-moment-perturbed} is an equality for $\Pcode$, and
Proposition~\ref{prop:lloyd-support} shows that its nonconstant dual spectrum is supported on
$\{6,9\}$, where $W_{11,3}=p_\tau$ (Appendix~\ref{app:finite-certificates}).
Hence \eqref{eq:spectral-discrepancy} gives
$\Dtwo(\Pcode)=7\,446\,692/3^{11}$.
The classical ternary Golay code is a perfect code with these parameters
\cite{MacWilliamsSloane1977}, so the lower bound is attained.
Conversely, if a code attains this benchmark, then
\eqref{eq:ternary-golay-stability} forces $\Phi_2(\C)=0$.
Because $NV_2=3^{11}$, Lemma~\ref{lem:defect-identities} shows that
$\C$ is perfect.
\end{proof}

On the admissible two-error parameter sets, define
\begin{equation}
\label{eq:all-q-two-constant}
 \sigma_{n,q}^{(2)}\triangleq
 \begin{cases}
  \vartheta_{n,q},&q\geq4,\\
  1,&(q,n)=(2,5),\\
  274/9,&(q,n)=(3,11).
 \end{cases}
\end{equation}

\begin{corollary}[All-alphabet two-error stability]
\label{cor:all-q-two-stability}
Let $G$ be a finite abelian group of order $q\geq2$, and let $n\geq5$.
Assume that $N\triangleq q^n/V_2$ is an integer and that
$L_2$ has two distinct integral roots in $\{1,\ldots,n\}$.  Then every
$N$-word code $\C\subseteq G^n$ satisfies
\begin{equation}
\label{eq:all-q-two-stability}
 \Dtwo(\C)-\delta_{n,q}^{(2)}
 \geq\sigma_{n,q}^{(2)}\frac{\Phi_2(\C)}{N^2}.
\end{equation}
If no perfect two-error code exists, then $\Dtwo(q,n,N)$ is at least
$\delta_{n,q}^{(2)}+4\sigma_{n,q}^{(2)}/N^2$.
\end{corollary}

\begin{proof}
For $q\geq4$, this is Theorem~\ref{thm:two-stability}.  For $q=2,3$,
Lemma~\ref{lem:small-q-two-parameters} leaves only $(2,5,2)$ and
$(3,11,3^6)$.  Theorem~\ref{thm:repetition} and Proposition
\ref{prop:ternary-golay-stability} give the estimates.  Direct substitution in
\eqref{eq:formal-masses} and \eqref{eq:two-benchmark} identifies their
left-hand benchmarks with $\delta_{5,2}^{(2)}$ and
$\delta_{11,3}^{(2)}$, respectively.  The gap statement follows from
Proposition~\ref{prop:transfer}.
\end{proof}

\section{The remaining perfect codes}
\label{sec:remaining-perfect}

The one-error theorem and the all-alphabet two-error theorem above cover every
perfect code correcting one or two errors, and Theorem~\ref{thm:repetition}
covers the binary repetition codes for every $e$.  It remains to record a
certificate for the binary Golay code and then invoke the known parameter
classification to complete the proof of
Theorem~\ref{thm:global-stability}.

The centered quartic used below is the exact-minimization certificate
recorded in \cite[Appendix~E]{Zabokritskiy2026Perfect}.  We include its
falling-factorial expansion, moment evaluation, symmetry argument, and
complete finite comparison.

\begin{proposition}[Binary Golay stability]
\label{prop:binary-golay-stability}
Every code $\C\subseteq\mathbb F_2^{23}$ of size $N=2^{12}$ satisfies
\begin{equation}
\label{eq:binary-golay-stability}
 \Dtwo(\C)-\frac{3\,277\,860\,456}{2^{23}}
 \geq\frac{4136}{25}\frac{\Phi_3(\C)}{N^2}.
\end{equation}
Every perfect three-error-correcting code attains the benchmark on the left.
Consequently,
\[
 \min_{\substack{\C\subseteq\mathbb F_2^{23}\\|\C|=2^{12}}}\Dtwo(\C)
 =\frac{3\,277\,860\,456}{2^{23}},
\]
and equality in this minimum holds precisely for perfect
three-error-correcting codes.
\end{proposition}

\begin{proof}
Here
\[
 V_3=1+\binom{23}{1}+\binom{23}{2}+\binom{23}{3}=2048
\]
and
\[
 c_w(3)=-\frac43(w-8)(w-12)(w-16).
\]
Define
\begin{equation}
\label{eq:binary-golay-polynomial}
 p(w)\triangleq705\,432-8\,992(w-12)^2+10\,000(w-12)^4.
\end{equation}
In the falling-factorial basis this is
\begin{equation}
\label{eq:binary-golay-falling}
\begin{aligned}
 p(w)={}&206\,770\,584-60\,743\,184w
       +7\,261\,008(w)_2\\
       &-420\,000(w)_3
       +10\,000(w)_4.
\end{aligned}
\end{equation}

Put
\[
 M_j\triangleq\sum_{w=1}^{23}(w)_j\Adual_w.
\]
The mass identity and Lemma~\ref{lem:dual-factorial-moments} give
\begin{align}
 M_0&=2047,\label{eq:binary-golay-m0}\\
 M_1&=23\,552-1\,024A_1,\label{eq:binary-golay-m1}\\
 M_2&=259\,072-22\,528A_1+1\,024A_2,
 \label{eq:binary-golay-m2}\\
 M_3&=2\,720\,256-354\,816A_1+32\,256A_2-1\,536A_3,
 \label{eq:binary-golay-m3}\\
 M_4&=27\,202\,560-4\,730\,880A_1+645\,120A_2
      -61\,440A_3+3\,072A_4.
 \label{eq:binary-golay-m4}
\end{align}
Substitution in \eqref{eq:binary-golay-falling} gives
\begin{equation}
\label{eq:binary-moment}
\begin{aligned}
 \sum_{w=1}^{23}\Adual_wp(w)
 ={}&3\,277\,860\,456
 +338\,952\,192(A_1+A_2)\\
 &+30\,720\,000(A_3+A_4)\\
 \geq{}&3\,277\,860\,456.
\end{aligned}
\end{equation}

We also record the symmetry used in the finite check.  When $q=2$,
equation \eqref{eq:W-coefficient-norm} reads
\[
 W_{23,2}(w)
 =\left\|(1+z)^{23-w}(1-z)^{w-1}\right\|_2^2.
\]
The substitution $z\mapsto-z$ interchanges the two factors and leaves the
coefficient norm unchanged.  Therefore
\[
 W_{23,2}(24-w)=W_{23,2}(w).
\]
The polynomial $p$ is invariant under the same reflection, and
$c_{24-w}(3)^2=c_w(3)^2$.  The exact values displayed in
Appendix~\ref{app:finite-certificates} consequently prove
\begin{equation}
\label{eq:binary-pointwise}
 W_{23,2}(w)-p(w)\geq\frac{4136}{25}c_w(3)^2
 \qquad(1\leq w\leq23).
\end{equation}
Summing \eqref{eq:binary-pointwise} against $\Adual_w$ and applying
\eqref{eq:spectral-discrepancy}, \eqref{eq:defect-parseval}, and
\eqref{eq:binary-moment} proves the result.

If $\Pcode$ is perfect, then $A_1(\Pcode)=\cdots=A_4(\Pcode)=0$, so
\eqref{eq:binary-moment} is an equality for $\Pcode$, and
Proposition~\ref{prop:lloyd-support} shows that its nonconstant dual spectrum is supported on
$\{8,12,16\}$, where $W_{23,2}=p$ (Appendix~\ref{app:finite-certificates}).
Hence \eqref{eq:spectral-discrepancy} gives
$\Dtwo(\Pcode)=3\,277\,860\,456/2^{23}$.
The classical binary Golay code is a perfect code with these parameters
\cite{MacWilliamsSloane1977}, so the lower bound is attained.
Conversely, equality at the benchmark in
\eqref{eq:binary-golay-stability} forces $\Phi_3(\C)=0$.
Since $NV_3=2^{23}$, Lemma~\ref{lem:defect-identities} shows that
$\C$ is perfect.
\end{proof}

\subsection{Completion of the global theorem}

\begin{proof}[Proof of Theorem~\ref{thm:global-stability}]
For every nontrivial perfect code correcting $e$ errors, sphere packing gives
$NV_e=q^n$, and distinct codewords are at distance at least $2e+1$, so
$n\geq2e+1$.  When $e=1$, \eqref{eq:dual-mass} gives
$\sum_{w\geq1}\Adual_w=V_1-1>0$, so Proposition~\ref{prop:lloyd-support}
forces the root $k=V_1/q$ to be integral.  When $e=2$, Lloyd's theorem
\cite{MacWilliamsSloane1977}, which holds over arbitrary alphabets
\cite{Delsarte1973}, gives two distinct integral roots of $L_2$ in the
Hamming weight range; Proposition~\ref{prop:lloyd-support} shows that the
nonconstant dual spectrum is supported on them.  Thus an actual
nontrivial perfect code satisfies the arithmetic hypotheses of the
corresponding stability theorem.

Theorem~\ref{thm:one-uniform} treats every one-error perfect code.
Corollary~\ref{cor:all-q-two-stability} treats every two-error parameter set.
Theorem~\ref{thm:repetition} treats the remaining odd binary repetition
codes, and Proposition~\ref{prop:binary-golay-stability} treats the binary
Golay parameters.

For prime-power alphabets, the classical parameter classification leaves
only the Hamming, ternary Golay, binary Golay, repetition, and trivial
families \cite{Tietavainen1973,ZinovievLeontiev1973}.  For non-prime-power
alphabets, the remaining cases correcting at least three errors are excluded by the
known nonexistence results: \cite{Reuvers1977} for $e\in\{3,4,5\}$,
\cite{Hong1984} for $e\in\{6,8\}$, and \cite{Best1983} for $e=7$ and $e\geq9$; the
radius-one and radius-two arguments above apply to every alphabet
size.  Trivial full-space and one-word codes
have no nonperfect competitors of the same cardinality and were excluded
from the statement.  The listed cases therefore exhaust all nontrivial
perfect codes.

In every branch, the preceding local calculations identify the benchmark with
$\Dtwo(\Pcode)$ and yield \eqref{eq:global-stability}.  Equality forces
$\Phi_e(\C)=0$, so $\C$ is perfect; conversely, every perfect code with the
same parameters attains the benchmark.  This proves both the
inequality and the equality characterization in the theorem.
\end{proof}

\section{Consequences}
\label{sec:consequences}

We now apply the defect identities of Section~\ref{sec:defect-main} to the
stability estimates.

\subsection{Smoothing bounds}

The global stability theorem and Proposition~\ref{prop:transfer} give the
following bounds.  For every nontrivial perfect code $\Pcode$ and every
competitor $\C$ of the same size,
\begin{align}
\label{eq:global-chi}
 \chi^2(P_{e,\C}\|U)
 &\leq\frac{N}{\sigma_{n,q,e}V_e}
 \bigl(\Dtwo(\C)-\Dtwo(\Pcode)\bigr),\\
\label{eq:global-tv}
 \TV(P_{e,\C},U)
 &\leq\frac{N}{2\sigma_{n,q,e}V_e}
 \bigl(\Dtwo(\C)-\Dtwo(\Pcode)\bigr).
\end{align}
By \eqref{eq:tv-holes}, the left side of \eqref{eq:global-tv} is the fraction
$H_e(\C)/q^n$ of holes.  Thus the discrepancy excess
controls both a smoothing error and the number of holes.

\subsection{Close pairs}

Let $P_i(\C)\triangleq NA_i(\C)$ be the number of ordered codeword pairs at
distance $i$.  Under
\eqref{eq:one-parameters}, \eqref{eq:phi-one-distance} reads
\[
 \Phi_1(\C)=qP_1(\C)+2P_2(\C).
\]
Consequently, Theorem~\ref{thm:one-uniform} gives the explicit bound
\[
 qP_1(\C)+2P_2(\C)
 \leq N^2\bigl(\Dtwo(\C)-\delta_{n,q}^{(1)}\bigr).
\]
In particular,
\[
 P_1(\C)\leq\frac{N^2}{q}
 \bigl(\Dtwo(\C)-\delta_{n,q}^{(1)}\bigr),
 \qquad
 P_2(\C)\leq\frac{N^2}{2}
 \bigl(\Dtwo(\C)-\delta_{n,q}^{(1)}\bigr).
\]
These estimates bound, for an arbitrary competitor, the number of ordered
codeword pairs whose radius-one balls overlap.

\subsection{Nearly-perfect covering profiles}

The preceding bounds use multiplicity one as the reference profile.  At
cardinalities where exact tiling is unavailable, a natural reference profile takes
only the two integer values nearest to the mean multiplicity.  Nearly-perfect covering codes provide
examples of this situation
\cite{BoruchovskyEtzionRoth2025Nearly,SacHimelfarbSchwartz2026Nearly}.

Let $\C$ be a binary covering code of length $n$, cardinality $N$, and
covering radius $R$, and put $\alpha\triangleq NV_R/2^n$, its mean covering
multiplicity.  For the nearly-perfect codes of
\cite[Definition~1]{SacHimelfarbSchwartz2026Nearly},
\cite[Lemma~9]{SacHimelfarbSchwartz2026Nearly} shows that every
over-covered word is covered by exactly two codewords.  Since these are
covering codes, $m_{R,\C}(x)\in\{1,2\}$ for every $x$ and
$1\leq\alpha\leq2$.  Among integer profiles of fixed mean, the centered sum
of squares is minimized by using only $\lfloor\alpha\rfloor$ and
$\lceil\alpha\rceil$.  Thus these codes attain
\[
 \sum_{x\in\mathbb F_2^n}\bigl(m_{R,\C}(x)-\alpha\bigr)^2
 =NV_R-2^n-\frac{(NV_R-2^n)^2}{2^n}.
\]
They therefore minimize the radius-$R$ discrepancy summand among binary
codes of the same length and cardinality.  Extending this comparison to
total discrepancy requires control of the other radii, which remains open.

\section{Conclusion}

For every nontrivial perfect code, the discrepancy excess, a
quantity summed over all radii, controls the tiling defect at the packing
radius with an explicit coefficient.  For one error, the uniform coefficient
one, which is sharp, shows that the discrepancy of a code over the radii other
than one is still at least that of a perfect code.  The parameter-dependent coefficients, the two-error estimates
under arithmetic conditions alone, and the bounds for the remaining perfect
families quantify this control.  The same estimate bounds the failure of
uniform ball noise to smooth the code distribution.

\paragraph{Proximity to a perfect code}
The next geometric question is whether a small tiling defect forces a
code to be close in symmetric difference to some perfect code.  Such a
stability theorem would strengthen the present estimates, much as
vertex-isoperimetric stability in the Hamming cube turns a small boundary
excess into proximity to a generalized Hamming ball
\cite{KeevashLong2020Vertex}.

\paragraph{Optimal constants}
Our coefficients follow from explicit positivity certificates.  In
particular, the parameter-dependent one-error coefficient $\kappa_{n,q}/q^2$
comes from a pointwise
spectral minorant and
moment constraints.  Optimizing over dual distributions realized by actual
codes may yield larger coefficients.  Determining these optimal constants
along the $q$-ary Hamming codes of growing length remains open.

\paragraph{Other reference profiles}
When a perfect code is unavailable, one can ask which multiplicity profile
should replace the constant multiplicity one.  The nearly-perfect covering codes
discussed above minimize the discrepancy at their covering radius; their
behavior under the total discrepancy is open.  For $q\geq4$,
Proposition~\ref{prop:distance-complete-monotone} expresses the discrepancy, up
to an additive constant and a positive factor, as an energy with a completely
monotonic potential; this is the setting of universal optimality in
\cite{CohnZhao2014Energy}.  Deleting one word from an
LP universally optimal code preserves universal optimality
\cite{CohnZhao2014Energy}; for $q\geq4$, the resulting code is another concrete test case for this question.

\paragraph{Larger radii and other spaces}
Extending the method to larger packing radii under arithmetic conditions
alone, or to other association schemes, calls for new positivity
certificates.  The spectral formulas for the discrepancy and the tiling defect,
which assign different weights to the same coordinates, suggest where to seek
them, while the required comparison depends on the ball geometry and the
available moment constraints.

\appendix

\section{Reproduced prerequisites}
\label{app:reproduced-prerequisites}

This appendix collects the proofs of statements from
\cite{Zabokritskiy2026Perfect} that are stated at their points of use.

\begin{proof}[Proof of Lemma~\ref{lem:ball-transform}]
Put $u(z)\triangleq1+(q-1)z$ and $b(z)\triangleq1-z$.  If $\wt(\xi)=w$,
coordinatewise character orthogonality gives
\[
 \sum_{x\in\X}\overline{\chi_\xi(x)}z^{\wt(x)}
 =u(z)^{n-w}b(z)^w.
\]
Thus the Fourier transform of the weight-$j$ sphere is
$K_j^{(n,q)}(w)$.  Summing its coefficients cumulatively,
\[
 \sum_{t\geq0}c_w(t)z^t
 =\frac1{1-z}\sum_{j=0}^nK_j^{(n,q)}(w)z^j
 =u(z)^{n-w}b(z)^{w-1}.
\]
The coefficient of $z^t$ is $K_t^{(n-1,q)}(w-1)$.  The right-hand
side has degree $n-1$, so $c_w(n)=0$.
\end{proof}

\begin{proof}[Proof of \eqref{eq:spectral-discrepancy} in
Proposition~\ref{prop:spectral-parseval}]
For fixed $t$, define
\[
 g_t\triangleq
 \frac1N\one_\C*\one_{B(0,t)}
 -\frac{V_t}{q^n}\one_{\X}.
\]
Thus
\[
 g_t(x)=\frac{|\C\cap B(x,t)|}{N}-\frac{V_t}{q^n}.
\]
Its Fourier transform vanishes at the trivial character.  At a nontrivial
frequency $\xi$ it is
\[
 \wh g_t(\xi)
 =\frac1N\wh{\one_\C}(\xi)c_{\wt(\xi)}(t).
\]
Plancherel, followed by summation over $t$, gives
\[
 \Dtwo(\C)
 =\frac1{q^nN^2}
   \sum_{\xi\neq0}|\wh{\one_\C}(\xi)|^2
   \sum_{t=0}^{n-1}c_{\wt(\xi)}(t)^2,
\]
where the radius-$n$ term vanishes by
Lemma~\ref{lem:ball-transform}.  Grouping frequencies by weight proves
\eqref{eq:spectral-discrepancy}.
\end{proof}

\begin{proof}[Proof of Proposition~\ref{prop:W-shape}]
Lemma~\ref{lem:ball-transform} shows that, for $0\leq t\leq n-1$, the
coefficient of $z^t$ in $u^{n-w}b^{w-1}$ is $c_w(t)$.  This proves
\eqref{eq:W-coefficient-norm}.  Parseval for polynomial coefficients gives
\[
 W_{n,q}(w)
 =\frac1{2\pi}\int_0^{2\pi}
 |u(e^{i\theta})|^{2(n-w)}
 |b(e^{i\theta})|^{2(w-1)}\,d\theta.
\]
On the unit circle,
\[
 |u(z)|^2-|b(z)|^2=qz^{-1}R_q(z).
\]
Taking the second finite difference under the integral and applying Parseval
again gives \eqref{eq:curvature-norm}; its right-hand side is positive because
the polynomial inside the norm is nonzero.  Finally,
\[
 |u(e^{i\theta})|^2-|b(e^{i\theta})|^2
 =q\bigl(q-2+2\cos\theta\bigr).
\]
For $q\geq4$ this is nonnegative and positive on a set of positive measure,
so the first-difference integral gives
$W_{n,q}(w+1)<W_{n,q}(w)$.
\end{proof}

\begin{proof}[Proof of Lemma~\ref{lem:dual-factorial-moments}]
Plancherel gives
\[
 \sum_{w=0}^n\Adual_w
 =\frac1{N^2}\sum_\xi|\wh{\one_\C}(\xi)|^2
 =\frac{q^n}{N}.
\]
The trivial frequency contributes $\Adual_0=1$, proving
\eqref{eq:dual-mass}.  For $j\geq1$, expand the Fourier squares over
ordered pairs:
\[
 \sum_{w=1}^n(w)_j\Adual_w
 =\frac1{N^2}\sum_{z,z'\in\C}
   \sum_\xi(\wt(\xi))_j\chi_\xi(z'-z).
\]
The falling factorial counts ordered $j$-tuples of distinct coordinates on
which $\xi$ is nontrivial.  If $d(z,z')=i$, character orthogonality makes
the inner sum vanish unless those $j$ coordinates contain all $i$ nonzero
coordinates of $z'-z$.  For $i\leq j$, the number of ordered choices is
\[
 (j)_i(n-i)_{j-i},
\]
and the character sum contributes
\[
 (-1)^i(q-1)^{j-i}q^{n-j}.
\]
There are $NA_i$ ordered codeword pairs at distance $i$.  Summing over
$0\leq i\leq j$ and dividing by $N^2$ proves
\eqref{eq:dual-factorial-moments}.
\end{proof}

\begin{proof}[Proof of the first three identities in
Lemma~\ref{lem:one-error-moments}]
Since $q^n/N=V=qk$ and $n(q-1)=V-1$, equations
\eqref{eq:dual-mass} and \eqref{eq:dual-factorial-moments} with $j=1$
give
\[
 \sum_{w=1}^n\Adual_w=V-1,
 \qquad
 \sum_{w=1}^nw\Adual_w=k(V-1-A_1).
\]
This proves \eqref{eq:mass-moment} and \eqref{eq:first-centered}.
The case $j=2$ gives
\[
 \sum_{w=1}^nw(w-1)\Adual_w
 =\frac{k}{q}\left(
 n(n-1)(q-1)^2
 -2(n-1)(q-1)A_1+2A_2\right).
\]
Using $q(k-1)=(n-1)(q-1)$ and expanding
$(w-k)^2=w(w-1)+(1-2k)w+k^2$ yields
\[
 \sum_{w=1}^n(w-k)^2\Adual_w
 =kA_1+\frac{2k}{q}A_2,
\]
which is \eqref{eq:second-centered}.
\end{proof}

\begin{proof}[Proof of Proposition~\ref{prop:distance-complete-monotone}]
Write
\[
 \mu_t(w)\triangleq|B(x,t)\cap B(y,t)|,\qquad
 \mu(w)\triangleq\sum_{t=0}^n\mu_t(w),
\]
where $d(x,y)=w$.  Expanding the squares in the definition of $\Dtwo$ gives
\[
 \Dtwo(\C)=
 \frac1N\sum_{i=0}^nA_i(\C)\mu(i)
 -\frac1{q^n}\sum_{i=0}^n\binom ni(q-1)^i\mu(i).
\]
For fixed $x,y,u$, the threshold identity
\[
\begin{aligned}
 \sum_{t=0}^n\bigl(
 &\one_{\{d(x,u)\leq t\}}+\one_{\{d(y,u)\leq t\}}\\
 &-2\one_{\{d(x,u)\leq t\}}\one_{\{d(y,u)\leq t\}}
 \bigr)
 =|d(x,u)-d(y,u)|
\end{aligned}
\]
shows that $\mu(0)-\mu(w)=\lambda(w)$.  The constant terms cancel in the
preceding expansion, proving \eqref{eq:distance-form}.  This is Barg's
invariance identity in the present normalization
\cite[Theorem~2.1]{Barg2021Stolarsky}.

For complete monotonicity, let $X_1,X_2,\ldots$ be independent with
\[
 \Pr(X_i=1)=\Pr(X_i=-1)=\frac1q,\qquad
 \Pr(X_i=0)=\frac{q-2}{q},
\]
and put $S_w\triangleq X_1+\cdots+X_w$.  Coordinatewise counting gives
\[
 \lambda(w)=\frac{q^n}{2}\E|S_w|,
 \qquad
 \lambda(w+1)-\lambda(w)=q^{n-1}p_w,
\]
where
\[
 p_w\triangleq\Pr(S_w=0)
 =\frac1{2\pi}\int_0^{2\pi}
 \left(1-\frac4q\sin^2\frac\theta2\right)^w\,d\theta.
\]
For $q\geq4$, the base of the power lies in $[0,1]$, and hence
\[
 (-1)^j\Delta^jp_w
 =\frac1{2\pi}\int_0^{2\pi}
 \left(1-\frac4q\sin^2\frac\theta2\right)^w
 \left(\frac4q\sin^2\frac\theta2\right)^j\,d\theta
 \geq0.
\]
The integral is positive whenever $j>0$.  Since
$\Delta f_{n,q}(w)=-q^{n-1}p_w$, all positive-order alternating differences
of $f_{n,q}$ are strictly positive; order zero follows from the monotonicity
of $\lambda$.
\end{proof}

\begin{proof}[Proof of Lemma~\ref{lem:small-q-two-parameters}]
Suppose first that $q=2$.  Since $V_2$ divides $2^n$, write
$V_2=2^m$.  The identity
\[
 8V_2=(2n+1)^2+7
\]
and the Ramanujan--Nagell theorem \cite{Nagell1961} leave
$n\in\{0,1,2,5,90\}$.  The two Lloyd roots are
\[
 \frac{n+1\pm\sqrt{n-1}}2.
\]
Thus $n=90$ is excluded by root integrality, whereas $n=5$ gives
$N=2$ and $(r,s)=(2,4)$.

Now suppose that $q=3$.  Here $V_2=2n^2+1$.  Writing $V_2=3^m$ gives
\[
 \frac{3^m-1}{2}=n^2.
\]
Ljunggren's theorem on $(X^m-1)/(X-1)=Y^2$
\cite{Ljunggren1943}, together with $n\geq5$, gives
$m=5$ and $n=11$.  Hence $N=3^6$, and direct factorization gives
\[
 L_2(w)=\frac92(w-6)(w-9).
\]
\end{proof}

\section{The ternary local-slope certificate}
\label{app:ternary-slope}

This appendix reproduces the exact ternary coefficient certificate from
\cite[Appendix~B]{Zabokritskiy2026Perfect}.  The resulting bound proves
\eqref{eq:local-slope} for $q=3$ and supplies the ternary case of
Proposition~\ref{prop:clean-spectral-gap}.
For $r\geq1$, define
\[
 \gamma_r\triangleq
 [z^{3r}](z-1)^{4r-3}(2z^2+5z+2)^r,
\]
where $[z^m]$ denotes the coefficient of $z^m$.

\begin{lemma}[Coefficient formula for the ternary slope]
\label{lem:ternary-coefficient}
For every $r\geq1$,
\begin{equation}
\label{eq:ternary-coefficient}
 W_{3r+1,3}(2r)-W_{3r+1,3}(2r+1)=6\gamma_r.
\end{equation}
\end{lemma}

\begin{proof}
Set
\[
 C_r(z)\triangleq(1+2z)^r(1-z)^{2r-1}.
\]
The coefficient-norm identity \eqref{eq:W-coefficient-norm}, specialized to
$q=3$, gives
\[
 W_{3r+1,3}(2r)=\|(1+2z)C_r\|_2^2,\qquad
 W_{3r+1,3}(2r+1)=\|(1-z)C_r\|_2^2.
\]
Since
\[
 |1+2z|^2-|1-z|^2=3(1+z+z^{-1})
\]
on the unit circle, their difference divided by $3$ is
\[
 [z^0](1+z+z^{-1})C_r(z)C_r(z^{-1}),
\]
where $[z^0]$ denotes constant-term extraction.  Now
\[
 C_r(z)C_r(z^{-1})
 =-z^{-(3r-1)}(z-1)^{4r-2}(2z^2+5z+2)^r.
\]
Using
\[
 1+z+z^{-1}=z^{-1}\frac{z^3-1}{z-1},
\]
the constant term is the negative of the difference between the
coefficients at degrees $3r-3$ and $3r$ in
\[
 H_r(z)\triangleq(z-1)^{4r-3}(2z^2+5z+2)^r.
\]
The polynomial $H_r$ is anti-reciprocal of degree $6r-3$, so these two
coefficients are $-\gamma_r$ and $\gamma_r$, respectively.  This proves
\eqref{eq:ternary-coefficient}.
\end{proof}

To prove positivity, we first convert $\gamma_r$ into a finite sum.
The Parseval representation used in the proof of
Proposition~\ref{prop:W-shape}, specialized to $n=3r+1$, $q=3$, and
$w=2r$, gives, with $y=\sin^2(\theta/2)$,
\[
\begin{aligned}
 &W_{3r+1,3}(2r+1)-W_{3r+1,3}(2r)\\
 &\qquad
 =\frac{3}{2\pi}\int_0^{2\pi}
 (9-8y)^r(4y)^{2r-1}(4y-3)\,d\theta.
\end{aligned}
\]
Together with \eqref{eq:ternary-coefficient}, and using the two-to-one
change of variables $y=\sin^2(\theta/2)$, this becomes
\[
 \gamma_r=-\frac1{2\pi}\int_0^1
 \frac{(9-8y)^r(4y)^{2r-1}(4y-3)}
      {\sqrt{y(1-y)}}\,dy.
\]
After $y=1-x$, expand $(1+8x)^r$.  Write $B(a,b)$ for the beta
function and $(a)^{\uparrow j}\triangleq a(a+1)\cdots(a+j-1)$, with
$(a)^{\uparrow0}\triangleq1$.  The
termwise integral is
\[
 \int_0^1x^{j-1/2}(1-x)^{2r-3/2}(1-4x)\,dx
 =B\left(j+\frac12,2r-\frac12\right)
  \frac{2r-3j-2}{2r+j}.
\]
Using
\[
 B\left(j+\frac12,2r-\frac12\right)
 =\frac{\pi}{4^{2r-1}}\binom{4r-2}{2r-1}
   \frac{(1/2)^{\uparrow j}}{(2r)^{\uparrow j}},
\]
we obtain
\begin{equation}
\label{eq:ternary-finite-sum}
 \gamma_r=\frac12\binom{4r-2}{2r-1}
 \sum_{j=0}^r\binom rj8^j
 \frac{(1/2)^{\uparrow j}}{(2r)^{\uparrow j}}
 \frac{3j-2r+2}{2r+j}.
\end{equation}

Define
\begin{align*}
 \mathsf P_0(r)&\triangleq486(r+1)(4r-1)(4r+1)(33r^2-4r-60),\\
 \mathsf P_2(r)&\triangleq2(r+2)(3r+4)(3r+5)(33r^2-70r-23),\\
 \mathsf T(r)&\triangleq104440+409424r+95070r^2-55242r^3+52272r^4.
\end{align*}
For $r\geq3$, all three quantities are positive.  For $\mathsf T(r)$, this is
seen by writing $x=r-3$:
\[
 \mathsf T(r)=4930840+5133686x+2420580x^2
       +572022x^3+52272x^4.
\]

\begin{lemma}[Telescoping recurrence]
\label{lem:ternary-telescoper}
For every $r\geq1$,
\begin{equation}
\label{eq:ternary-positive-recurrence}
 16\mathsf P_2(r)(\gamma_{r+2}-16\gamma_{r+1})
 =\mathsf P_0(r)(\gamma_{r+1}-16\gamma_r)+\mathsf T(r)\gamma_{r+1}.
\end{equation}
\end{lemma}

\begin{proof}
Define
\[
 b_{s,j}\triangleq
 \frac12\binom{4s-2}{2s-1}\binom sj8^j
 \frac{(1/2)^{\uparrow j}}{(2s)^{\uparrow j}},\qquad
 a_{s,j}\triangleq b_{s,j}\frac{3j-2s+2}{2s+j},
\]
and set both quantities to zero outside $0\leq j\leq s$.  Thus
\eqref{eq:ternary-finite-sum} says
\[
 \gamma_s=\sum_{j=0}^sa_{s,j}.
\]
Let
\[
 \mathsf P_1(r)\triangleq
 -16\mathsf P_2(r)-\frac{\mathsf P_0(r)+\mathsf T(r)}{16},
\]
and put
\[
 \delta_r\triangleq16r^3+64r^2+79r+30.
\]
Define
\begin{align*}
m_0(r)\triangleq{}&-9240-107456r-85526r^2+188518r^3\\*
&\quad+231594r^4+45738r^5-13068r^6,\\
m_1(r)\triangleq{}&-48900-67930r+382778r^2+823841r^3\\*
&\quad+601887r^4+211266r^5+39204r^6,\\
m_2(r)\triangleq{}&-200880-384462r+17313r^2+314580r^3\\*
&\quad+143550r^4+13068r^5,\\
m_3(r)\triangleq{}&-9180+144648r+145953r^2-71145r^3-72171r^4,\\
m_4(r)\triangleq{}&59940+91476r-27135r^2-48114r^3,\\
m_5(r)\triangleq{}&14580+972r-8019r^2.
\end{align*}
Finally, put
\[
 H(r,j)\triangleq j\sum_{\ell=0}^5m_\ell(r)j^\ell,
 \qquad
 \mathcal R(r,j)\triangleq\frac{H(r,j)}{24\delta_r}.
\]
The following rational identity is the required telescoping certificate.
Exact simplification gives, for $r\geq1$ and $0\leq j\leq r+2$,
\begin{equation}
\label{eq:ternary-telescoper}
\begin{aligned}
 &\mathsf P_0(r)a_{r,j}+\mathsf P_1(r)a_{r+1,j}
  +\mathsf P_2(r)a_{r+2,j}\\
 &\qquad=b_{r+2,j+1}\mathcal R(r,j+1)
             -b_{r+2,j}\mathcal R(r,j).
\end{aligned}
\end{equation}
For a direct check, divide by $b_{r+2,j}$, clear denominators, and use
\[
 \frac{b_{s,j+1}}{b_{s,j}}
 =\frac{4(s-j)(2j+1)}{(j+1)(2s+j)},\qquad
 \frac{b_{s-1,j}}{b_{s,j}}
 =\frac{(s-j)(2s+j-2)(2s+j-1)}
        {4s(4s-5)(4s-3)}.
\]
After substitution, the difference between the two sides of
\eqref{eq:ternary-telescoper} is the zero polynomial in $r,j$.
At the boundaries, $\mathcal R(r,0)=0$ and $b_{r+2,r+3}=0$.
Summing over $j$ proves
\[
 \mathsf P_0(r)\gamma_r+\mathsf P_1(r)\gamma_{r+1}
 +\mathsf P_2(r)\gamma_{r+2}=0.
\]
Multiplication by $16$ and the definition of $\mathsf P_1$ give
\eqref{eq:ternary-positive-recurrence}.
\end{proof}

\begin{lemma}[Positivity of the ternary slope coefficient]
\label{lem:ternary-recurrence}
For every $r\geq1$, the integer $\gamma_r$ satisfies $\gamma_r\geq2$.
\end{lemma}

\begin{proof}
Direct coefficient extraction gives
\[
 \gamma_1=2,\quad \gamma_2=15,\quad \gamma_3=188,\quad
 \gamma_4=2977,\quad \gamma_5=54468,
\]
and $\gamma_5-16\gamma_4=6836>0$.  Since
$\mathsf P_0(r),\mathsf P_2(r),\mathsf T(r)>0$ for
$r\geq3$, \eqref{eq:ternary-positive-recurrence}, applied successively
for $r=4,5,\ldots$, gives
\[
 \gamma_{r+1}>16\gamma_r>0\qquad(r\geq4).
\]
Together with the displayed initial values, this proves
$\gamma_r\geq2$ for every $r\geq1$.
\end{proof}

\section{Formal two-error quadrature and comparison}
\label{app:formal-comparison}

The two-node quadrature and the formal Delsarte comparison proved below are,
respectively, Lemma~A.1 and Lemma~5.3 of
\cite{Zabokritskiy2026Perfect}; the quadrature is the strength-four
Rao-endpoint case of Levenshtein's quadrature
\cite[Section~2.5]{BoyvalenkovEtAl2017}, and the comparison is the
pair-covering certificate of \cite[Proposition~21]{CohnZhao2014Energy}
applied to a formal reference.  We reproduce their argument in the notation of
the present paper.  The subtraction of $c\mu_2$ in \eqref{eq:formal-perturbation} and the
coefficientwise regrouping in Section~\ref{sec:higher-error} are new.

\begin{proof}[Proof of Lemma~\ref{lem:formal-two-comparison}]
A direct expansion of the Lloyd polynomial around the binomial
mean gives
\begin{equation}
\label{eq:lloyd-two-centered}
 2L_2(\bar w+y)=q^2y^2+q(q-4)y+2-(q-1)n.
\end{equation}
Thus $L_2(\bar w)<0$, whereas $L_2(1)>0$ and $L_2(n)>0$.
The integral roots therefore satisfy
$1<r<\bar w<s<n$.  Their sum and separation are
\begin{equation}
\label{eq:root-sum-difference}
 r+s=\left(2-\frac2q\right)(n-2)+3,
 \qquad
 (s-r)^2=1+\frac{4(q-1)(n-2)}{q^2}.
\end{equation}
In particular, $s-r\geq2$.  Combining the root sum with
\eqref{eq:lloyd-two-centered} gives
\[
 s-\bar w
 =\frac{s-r}{2}-\frac{q-4}{2q}
 \geq\frac{q+4}{2q}>\frac1q
 \geq\frac{\bar w}{V_2-1}.
\]
The definitions in \eqref{eq:formal-masses} now give $b_r,b_s>0$.

We next establish the quadrature identity
\begin{equation}
\label{eq:radius-two-quadrature}
 \frac{P(0)+b_rP(r)+b_sP(s)}{V_2}
 =\sum_{i=0}^nq^{-n}\binom ni(q-1)^iP(i)
 \qquad(\deg P\leq4).
\end{equation}
The two sides agree on $1$ and $x$ by the definitions of $b_r,b_s$.
They agree on $L_2$ because
\[
 L_2=K_0^{(n,q)}+K_1^{(n,q)}+K_2^{(n,q)},\qquad
 L_2(0)=V_2,\qquad L_2(r)=L_2(s)=0,
\]
and the binomial average of $L_2$ is one.  Hence they agree on every
quadratic.  Every polynomial of degree at most four can be written as
\[
 P(x)=R(x)+xL_2(x)S(x),\qquad \deg R\leq2,\quad \deg S\leq1.
\]
The second term vanishes at $0,r,s$.  Its binomial average also vanishes:
using $x\binom nx=n\binom{n-1}{x-1}$, this is the orthogonality of
$L_2(x)=K_2^{(n-1,q)}(x-1)$ against polynomials of degree at most one.
This proves \eqref{eq:radius-two-quadrature}.  Applying it to
$K_i^{(n,q)}$ for $0\leq i\leq4$ gives $A_0^*=1$ and
$A_1^*=\cdots=A_4^*=0$.  Finally, $Q_0^*=1$ and the transform normalization
give $\sum_iA_i^*=N$, proving \eqref{eq:A-star-low}.

It remains to prove the comparison.  Let $A$ be the distance distribution of
an $N$-word code, and define its MacWilliams--Delsarte transform by
\[
 Q_j\triangleq\frac1N\sum_{i=0}^nA_iK_j^{(n,q)}(i).
\]
Then $Q_0=1$, every $Q_j$ is nonnegative, and
$\mathsf KQ=V_2A$.  For a completely monotonic function $u$, put
$g\triangleq\mathsf K^{\mathsf T}u$.  The map $u\mapsto\mathsf K^{\mathsf T}u$ preserves complete
monotonicity \cite[Lemma~10]{CohnZhao2014Energy}.  The Krawtchouk
polynomials, the matrix $\mathsf K$, and complete monotonicity on
$\{0,\ldots,n\}$ have the same normalization and definition in
\cite{CohnZhao2014Energy} as here, and in \cite[Section~I]{CohnZhao2014Energy},
$\mathbb F_q$ denotes an arbitrary alphabet with $q$ elements and no field
structure is used; hence the results of \cite{CohnZhao2014Energy} used in this
proof hold for every integer $q\geq2$.  Let $h$ be the cubic
polynomial interpolating $g$ at
\[
 r-1,\quad r,\quad s,\quad s+1.
\]
These nodes form a pair covering in the sense of
\cite[Definition~20]{CohnZhao2014Energy}; they consist of the two pairs
$\{r-1,r\}$ and $\{s,s+1\}$ of consecutive integers, so
$(x-r+1)(x-r)(x-s)(x-s-1)\geq0$ at every integer $x$.  Write
$g[x_0,\ldots,x_j]$ for divided differences.  At integer nodes, a divided
difference of order $j$ is a nonnegative combination of values of
$\Delta^jg/j!$, so complete monotonicity gives it the sign of $(-1)^j$.  At an
integer $i$ other than a node, $g(i)-h(i)$ is $g[s+1,s,r,r-1,i]$ times the
nonnegative product above, and, with the nodes in the order $s+1,s,r,r-1$, the
coefficients of $h$ in the Newton basis below are $g(s+1)$, $-g[s+1,s]$,
$g[s+1,s,r]$, and $-g[s+1,s,r,r-1]$.  This is the interpolation theorem of
\cite[Lemma~19]{CohnZhao2014Energy}: it yields
$h(i)\leq g(i)$ for $1\leq i\leq n$ and a Newton expansion with nonnegative
coefficients in
\[
 1,\qquad s+1-x,\qquad (s+1-x)(s-x),\qquad
 (s+1-x)(s-x)(r-x).
\]
Each nonconstant function in this list is positive definite, that is, has
nonnegative Krawtchouk coefficients \cite[Definition~11]{CohnZhao2014Energy}.  Since
$K_1^{(n,q)}(x)=q(\bar w-x)$ and $s>\bar w$, the functions $s+1-x$ and $s-x$
are positive combinations of $K_0^{(n,q)}$ and $K_1^{(n,q)}$.  Hence the first
case holds, the second follows by \cite[Lemma~12]{CohnZhao2014Energy} (as in
\cite[Lemma~13 and Corollary~14]{CohnZhao2014Energy}), and the cubic case
follows from
\[
 (s+1-x)(s-x)(r-x)=\frac{2}{q^2}(s+1-x)L_2(x),
\]
because $s+1-x$ is positive definite, $L_2=K_0^{(n,q)}+K_1^{(n,q)}+K_2^{(n,q)}$,
and products of positive definite functions are positive definite
\cite[Lemma~12]{CohnZhao2014Energy}.  Consequently,
\[
 h(x)=\sum_{j=0}^3\eta_jK_j^{(n,q)}(x),
 \qquad \eta_j\geq0\quad(j\geq1).
\]

Using $A_j\geq0$, $A_0=1$, and $\mathsf KQ=V_2A$, we obtain
\begin{align*}
 \sum_{i=1}^ng(i)Q_i
 &\geq\sum_{i=1}^nh(i)Q_i\\
 &=V_2\sum_{j=0}^3\eta_jA_j-h(0)\\
 &\geq V_2\eta_0-h(0).
\end{align*}
For $Q^*$, the only nonzero coordinates among $1,\ldots,n$ are $r$ and $s$,
where $h=g$, and \eqref{eq:A-star-low} gives equality with
$V_2\eta_0-h(0)$.  Adding the common coordinate at zero yields
$g^{\mathsf T}Q\geq g^{\mathsf T}Q^*$.  Since
\[
 A=\frac1{V_2}\mathsf KQ,
 \qquad A^*=\frac1{V_2}\mathsf KQ^*,
\]
we conclude that $u^{\mathsf T}(A-A^*)\geq0$, as required.
\end{proof}

\section{Evaluation of the two-error coefficient}
\label{app:vartheta-evaluation}

We first prove the arithmetic implication used in the proof of
Theorem~\ref{thm:two-stability}.  Let $r<s$ be the two integral Lloyd roots.
The root identities in \eqref{eq:root-sum-difference} imply
$q\mid2(n-2)$.  If $q$ is even, the first two
possible positive values of $n-2$ are $q/2$ and $q$; the second expression
in \eqref{eq:root-sum-difference} is then $3-2/q$ and $5-4/q$,
respectively.  The former is never an integral square.  The latter is not an
integral square unless $q=4$; in that case $r+s=9$ and $s-r=2$ have
opposite parity, so the roots are not integral.  The next possible value is
$n-2=3q/2$, which already gives $n\geq q+3$.  If $q$ is odd, the first
possible value is $n-2=q$, and again $5-4/q$ is not an integral square.
The next possible value is $n-2=2q$, which again gives $n\geq q+3$.
Thus the stated hypotheses force
\begin{equation}
\label{eq:admissible-length}
 n\geq q+3.
\end{equation}

For the coefficient evaluation below, assume only $q\geq4$ and $n\geq5$.
We now prove Proposition~\ref{prop:vartheta-closed}.  Write $f\triangleq f_{n,q}$ and
put $m\triangleq n-1$.  Recall the quantities $F_t$ and $M_t$ defined for
$1\leq t\leq n-1$ in Section~\ref{sec:two-error}, and extend the same
definitions to $t=0,n$:
\[
 F_t\triangleq\mathcal D_{n-t,t}f,
 \qquad
 M_t\triangleq\mathcal D_{n-t,t}\mu_2
 \qquad(0\leq t\leq n).
\]
The recurrence
\[
 \mathcal D_{i,w}u
 =\mathcal D_{i,w+1}u+\mathcal D_{i+1,w}u
\]
and Pascal's identity give the finite-difference identity
\cite[Lemma~4]{CohnZhao2014Energy}
\begin{equation}
\label{eq:right-newton-general}
 \mathcal D_{i,w}u
 =\sum_{\ell=0}^{n-w-i}\binom{n-w-i}{\ell}
   \mathcal D_{i+\ell,n-i-\ell}u.
\end{equation}
Taking $i=0$ recovers the Newton expansion
\eqref{eq:right-newton} used in Section~\ref{sec:two-error}.  Thus $u$ is
completely monotonic if and only if the differences $\mathcal D_{n-t,t}u$
are nonnegative for every $0\leq t\leq n$.
Since $\mu_2$ vanishes beyond distance four, $M_t=0$ for
$5\leq t\leq n$.  The corresponding $F_t$ are positive for
$5\leq t\leq n-1$, while $F_n=f(n)=0$.

For $0\leq w\leq m$, define
\[
 p_w\triangleq\frac1{2\pi}\int_0^{2\pi}
 \left(1-\frac4q\sin^2\frac\theta2\right)^w\,d\theta.
\]
The increment calculation in the proof of
Proposition~\ref{prop:distance-complete-monotone} in Appendix~\ref{app:reproduced-prerequisites} gives
\[
 \Delta f(w)=-q^m p_w
 \qquad(0\leq w\leq m).
\]
Since $F_t=q^m(-1)^{m-t}\Delta^{m-t}p_t$, we have
$F_t=\frac1{2\pi}\int_0^{2\pi}(q-4y)^t(4y)^{m-t}\,d\theta$ with
$y=\sin^2(\theta/2)$.  Writing $q-4y=(q-4)+4(1-y)$, expanding, and using the
beta integral gives, for $0\leq t\leq m$,
\begin{equation}
\label{eq:F-diagonal}
 F_t=\binom{2(m-t)}{m-t}
 \sum_{k=0}^t\binom tk(q-4)^{t-k}
 \frac{\binom{2k}{k}}{\binom{m-t+k}{k}}.
\end{equation}
Taking $t=4$ in \eqref{eq:F-diagonal} and writing $a\triangleq q-4$ gives
\begin{equation}
\label{eq:F-four-evaluation}
 F_4=
 \frac{\binom{2m-8}{m-4}\psi_4(m,a)}
 {m(m-1)(m-2)(m-3)}.
\end{equation}
Comparing \eqref{eq:F-four-evaluation} with \eqref{eq:vartheta-closed} proves
\eqref{eq:F4-vartheta}; since $M_4=6$ by \eqref{eq:M-diagonal}, equivalently
$\vartheta_{n,q}=F_4/M_4$.
The fixed-$q$ asymptotics in \eqref{eq:vartheta-asymptotics} follow from
\[
 \binom{2m-8}{m-4}
 \sim\frac{4^n}{1024\sqrt{\pi n}}
\]
and the leading term of \eqref{eq:P-four}.  This completes the proof of
Proposition~\ref{prop:vartheta-closed} and supplies the formulas used in
Theorem~\ref{thm:two-stability}.

\section{Golay benchmarks and quantitative finite certificates}
\label{app:finite-certificates}

\subsection{Reproduced benchmark evaluations}

The exact benchmark calculations in this subsection appeared in the Golay
appendices of \cite{Zabokritskiy2026Perfect}.  They evaluate $\Dtwo(\Pcode)$ directly from the dual distribution of a perfect code
and agree with the values obtained in the proofs of Propositions~\ref{prop:ternary-golay-stability}
and \ref{prop:binary-golay-stability}.

For the ternary parameters, let $\Pcode$ be perfect.
Proposition~\ref{prop:lloyd-support} restricts its nonconstant dual spectrum
to the zeros $6,9$ of $c_w(2)$.  Since its minimum distance is at least five,
$A_1=A_2=0$.  The mass and first-moment identities therefore give
\[
 \Adual_6+\Adual_9=242,
 \qquad
 6\Adual_6+9\Adual_9=1782,
\]
so $\Adual_6=132$ and $\Adual_9=110$.  Direct evaluation gives
\[
 W_{11,3}(6)=30\,436,\qquad W_{11,3}(9)=31\,174,
\]
and consequently
\[
 \Dtwo(\Pcode)
 =\frac{132\cdot30\,436+110\cdot31\,174}{3^{11}}
 =\frac{7\,446\,692}{3^{11}}.
\]

For the binary parameters, let $\Pcode$ be perfect.  Its minimum distance is
at least seven, so $A_1=\cdots=A_4=0$.
Proposition~\ref{prop:lloyd-support} restricts its nonconstant dual spectrum
to $w=8,12,16$.  Equations~\eqref{eq:binary-golay-m0}--\eqref{eq:binary-golay-m2} give
\[
 (\Adual_8,\Adual_{12},\Adual_{16})=(506,1288,253).
\]
Moreover,
\[
 p(8)=p(16)=3\,121\,560,\qquad p(12)=705\,432,
\]
and the finite certificate has $W_{23,2}=p$ at these three weights.  Thus
\[
 \Dtwo(\Pcode)
 =\frac{506\cdot3\,121\,560+
        1288\cdot705\,432+
        253\cdot3\,121\,560}{2^{23}}
 =\frac{3\,277\,860\,456}{2^{23}}.
\]

\subsection{Quantitative ratio tables}

All entries below follow by direct exact evaluation of $W_{n,q}(w)$, the
displayed certificate polynomials, and the Lloyd polynomials $c_w(2)$ and $c_w(3)$.  Thus the tables
are part of the proof; the archived verifier
\cite{Zabokritskiy2026StabilityVerifier} provides an independent
reproduction of the same arithmetic.  The original certificate polynomials
appeared in Appendices~D and E of \cite{Zabokritskiy2026Perfect}; the ratio
tables and their minimizing values are used for the quantitative bounds proved
in the present paper.

For the ternary Golay parameters, using the polynomial $p_\tau$ from
\eqref{eq:ternary-perturbed-polynomial}, put
\[
 R_3(w)\triangleq
 \frac{W_{11,3}(w)-p_\tau(w)}{c_w(2)^2}
 \qquad(w\notin\{6,9\}).
\]
By \eqref{eq:W-coefficient-norm}, the values $W_{11,3}(1),\ldots,W_{11,3}(11)$
are $653\,808\,673$, $17\,343\,526$, $1\,546\,342$, $261\,424$, $71\,878$,
$30\,436$, $19\,756$, $20\,020$, $31\,174$, $68\,068$, and $184\,756$.
The exact values of $R_3$ are
\[
\begin{array}{c|r@{\qquad}c|r}
 w&R_3(w)&w&R_3(w)\\ \hline
 1&72\,612\,403/3600&2&317\,159/294\\
 3&17\,236/81&4&3656/45\\
 5&46&7&274/9\\
 8&92/3&10&274/9\\
 11&638/15&&
\end{array}
\]
At the Lloyd roots,
\[
 W_{11,3}(6)=p_\tau(6)=30\,436,\qquad
 W_{11,3}(9)=p_\tau(9)=31\,174,
\]
and $c_6(2)=c_9(2)=0$.  Hence the minimum of $R_3(w)$ is
$274/9$, attained at $w=7,10$, proving
\eqref{eq:ternary-pointwise}.  The coefficient $1\,233$ in
\eqref{eq:ternary-perturbed-polynomial} is the value at which these two
ratios coincide: increasing it raises $R_3(7)$ and lowers $R_3(10)$.

For the binary Golay parameters, the reflection symmetry from the proof of Proposition~\ref{prop:binary-golay-stability}
reduces the check to $1\leq w\leq12$.  With $p$ as in
\eqref{eq:binary-golay-polynomial}, put
\[
 R_2(w)\triangleq
 \frac{W_{23,2}(w)-p(w)}{c_w(3)^2}
 \qquad(w\notin\{8,12,16\}).
\]
Here $W_{23,2}(w)=\binom{46-2w}{23-w}\binom{2w-2}{w-1}\big/\binom{22}{w-1}$,
the case $n=23$ of the binary formula in the proof of
Proposition~\ref{prop:clean-spectral-gap}.
For the nonroot weights in the reduced range,
\[
\begin{array}{c|r@{\qquad}c|r}
 w&R_2(w)&w&R_2(w)\\ \hline
 1&26\,299\,411\,704/29\,645&2&190\,752\,843/4900\\
 3&1\,126\,552/195&4&25\,509/16\\
 5&558\,648/847&6&7023/20\\
 7&4952/25&9&1608/7\\
 10&731/4&11&4136/25
\end{array}
\]
At $w=8,12$, both $W_{23,2}(w)-p(w)$ and $c_w(3)$ vanish; reflection
gives the same statement at $w=16$.  The minimum of $R_2(w)$ is
$4136/25$ at $w=11$, and by reflection also at $w=13$.  This proves
\eqref{eq:binary-pointwise}.

\section*{Computational verification}
The exact-arithmetic
verifier archived on Zenodo, version~1.0.3, independently reproduces the
finite Golay tables, factorial-moment coefficients, two-error coefficient
identities, and bounded parameter search listed in its verification scope
\cite{Zabokritskiy2026StabilityVerifier}.

\begingroup
\hbadness=2000
\bibliographystyle{plainurl}
\bibliography{references}
\endgroup

\end{document}